\documentclass[final]{elsart}

\usepackage{multirow}
\usepackage{graphics}
\usepackage{amsmath}
\usepackage{tikz}
\usepackage[ruled]{algorithm}
\usepackage{amsfonts}
\usepackage{amssymb}

\theoremstyle{plain}
\newtheorem{corollary}{Corollary}
\newtheorem{theorem}{Theorem}
\newtheorem{proposition}{Proposition}
\newtheorem{lemma}{Lemma}

\theoremstyle{definition}
\newtheorem{definition}{Definition}
\newtheorem{example}{Example}

\newenvironment{proof}{{\bf Proof:~}}{}

\newcommand{\algline}[1]{line {\sc #1}}

\newcommand{\sepbar}{\mathrel{|}}
\newcommand{\sepdot}{\mathpunct{.}}

\newcommand{\infer}[3][]{#1\mathord{:}~\genfrac{}{}{0.5pt}{}{\begin{aligned}#2\end{aligned}}{#3}}
\newcommand{\jd}{\ensuremath{\mathop{\mathord{\Join}\mathord{:}}}}
\newcommand{\fd}{\ensuremath{\mathop{:}}}
\newcommand{\jdrightarrow}{
  \mathrel{
    \vcenter{\offinterlineskip
      \hbox{\hskip0.35ex\fontsize{9}{9}$\Join$}
      \hbox{$\rightarrow$}
    }
  }
}

\def\qed {{                
   \parfillskip=0pt        
   \widowpenalty=10000     
   \displaywidowpenalty=10000  
   \finalhyphendemerits=0  
   \leavevmode             
   \unskip                 
   \nobreak                
   \hfil                   
   \penalty50              
   \hskip.2em              
   \null                   
   \hfill                  
   $\square$
   \par}}                  

\newcommand{\rrc}{\color{green!10!red!65!black}}
\newcommand{\rrmark}[1]{%
\marginpar{\rrc $\bigstar$\\%
\raggedright \fontsize{8.5}{9.8}\selectfont#1}}
\newcommand{\rrtext}[1]{{\rrc#1}}

\renewcommand{\rrmark}[1]{}
\renewcommand{\rrtext}[1]{#1}

\journal{Information Systems}

\begin{document}
\begin{frontmatter}
\title{Consistent Query Answers\\ in the Presence of 
       Universal Constraints\thanksref{grants}}
\thanks[grants]{Research partially supported by NSF grants 
                IIS-0119186 and IIS-0307434 and 
                Enumeration project ANR-07-blanc}
  
\author{S\l{}awomir Staworko\corauthref{phd}}
\address{
  INRIA Lille - Nord Europe,
  Parc Scientifique de la Haute Borne,
  Park Plaza - B\^at A - 40 avenue Halley,
  59650 Villeneuve d'Ascq.
  Email: {\tt slawomir.staworko@inria.fr}
}
\corauth[phd]{Corresponding author. Part of this research was done
              when the author was a PhD student at 
              the University at Buffalo}  

\author{Jan Chomicki}
\address{
  Department of Computer Science and Engineering, 
  201 Bell Hall, The State University of New York at Buffalo, 
  Buffalo, NY 14260. 
  Email: {\tt chomicki@cse.buffalo.edu}
}

\begin{abstract}
The framework of {\em consistent query answers} and {\em repairs} has been introduced to alleviate the impact of inconsistent data on the answers to a query. A repair is a minimally different consistent instance and an answer is consistent if it is present in {\em every} repair. In this article we study the complexity of consistent query answers and {\em repair checking} in the presence of {\em universal constraints}. 

\noindent
We propose an extended version of the conflict hypergraph which allows to capture all repairs w.r.t. a set of universal constraints. We show that repair checking is in PTIME for the class of full tuple-generating dependencies and denial constraints, and we present a polynomial {\em repair algorithm}. This algorithm is {\em sound}, i.e. always produces a repair, but also {\em complete}, i.e. every repair can be constructed. Next, we present a polynomial-time algorithm computing consistent answers to ground quantifier-free queries in the presence of denial constraints, join dependencies, and acyclic full-tuple generating dependencies. Finally, we show that extending the class of constraints leads to intractability. For arbitrary full tuple-generating dependencies consistent query answering becomes coNP-complete. For arbitrary universal constraints consistent query answering is $\Pi_2^p$-complete and repair checking coNP-complete. 
\end{abstract}
\end{frontmatter}
\section{Introduction}
Traditionally, the consistency of a database with a set of integrity constraints was maintained by a DBMS~\cite{RaGe00}. While integrity constraints continue to express important properties of the stored data, in many novel database applications enforcing the consistency becomes problematic. For example, in the scenario of data integration even if the sources are separately consistent, together they may contribute conflicting data. Because data sources are often autonomous and their contents cannot be altered, the consistency cannot be restored by means of data manipulation. Consistency violations occur naturally also in the context of long running data manipulations, delayed updates on data warehouses, and legacy databases. Finally, consistency enforcement may be deactivated for efficiency reasons. At the same time, the semantic properties expressed by integrity constraints often influence the way the user formulates her queries. Hence, if the database is inconsistent, evaluating the queries may yield incorrect and misleading answers.

To address the problem of the potential impact of inconsistencies on  query results Arenas et al. have proposed the framework of repairs and consistent query answers~\cite{ArBeCh99}. A {\em repair} is a consistent database instance minimally different from the original one. The {\em consistent query answers} are the answers present in every repair. Intuitively, the repairs represent all possible ways to restore consistency in the database and an answer is consistent if it is obtained regardless of the way the conflicts are resolved, i.e. the answer that is not affected by the inconsistencies. This framework has served as a foundation for most of the subsequent work in the area of querying inconsistent databases (for the surveys of the area, see~\cite{BeCh03,Be06,ChMa04,Ch07,Fa08}).
\begin{example}
We consider a database that stores information on the occurrence of a genetically inherited disease neurofibromatosis (NF) causing tumors of the nervous tissue. 
NF is an autosomal dominant disorder, which means that only one mutated gene needs to be present in the genome of an affected person. Typically, this gene is inherited from one of the parents.\footnote{A spontaneous mutation can also take place, but we ignore it for sake of simplicity.}

The schema of the database contains two relations: $NF(\underline{Name}, Diag)$ and $Parent(Name,Child)$, where the underline indicates the (primary) key of a relation. The following constraint captures the inheritance factor of NF:
\begin{multline}\label{eq:nf-inh}
NF(x,`yes`)\land Parent(y_1,x)\land Parent(y_2,x)\land{}
y_1\neq y_2\\ 
\Rightarrow NF(y_1,`yes`) \lor NF(y_2,`yes`).
\end{multline}

Now, consider the database instance $I$ in Figure~\ref{fig:nf-instance}.
\begin{figure}[htb]
  \begin{center}
    \begin{tabular}{|l|l|}
      \multicolumn{2}{c}{\it NF} \\
      \hline
      \underline{\it Name} & {\it Diag}\\
      \hline
      Steve & no\\
      \hline
      Mary & no\\
      \hline
      Donald & yes\\
      \hline
    \end{tabular}
    \hspace{35pt}
    \begin{tabular}{|l|l|}
      \multicolumn{2}{c}{\it Parent}\\
      \hline
      \it Name & {\it Child}\\
      \hline
      Steve & Donald\\
      \hline
      Mary & Donald\\
      \hline
    \end{tabular}
  \end{center}
  \caption{\label{fig:nf-instance} Inconsistent database $I_1$.}
\end{figure}
This instance violates \eqref{eq:nf-inh}: $Donald$ is diagnosed with NF while neither of his parents are. This violation can be resolved in three ways: 
\begin{enumerate}
\item By inserting a tuple with a positive diagnosis for one of the $Donald$'s parents. Because of the key dependency, this creates a conflict with the already existing tuple which is consequently deleted. This yields the following repairs:
\begin{align*}
  I_1'=\{
  &NF(Steve,yes), NF(Mary,no), NF(Donald,yes),\\
  &Parent(Steve,Donald), Parent(Mary,Donald)\}.\\
  I_2'=\{
  &NF(Steve,no), NF(Mary,yes), NF(Donald,yes),\\
  &Parent(Steve,Donald), Parent(Mary,Donald)\}.
\intertext{
\item By removing one of the tuples of $Parent$ relation, which gives the following repairs:
}
  I_3'=\{
  &NF(Steve,no), NF(Mary,no), NF(Donald,yes),\\
  &Parent(Mary,Donald)\}.\\
  I_4'=\{
  &NF(Steve,no), NF(Mary,yes), NF(Donald,yes),\\
  &Parent(Steve,Donald)\}.
\intertext{
\item By removing the tuple with the diagnosis of $Donald$ giving the following repair:
}  I_5'=\{&NF(Steve,no), NF(Mary,no),\\
         &Parent(Steve,Donald), Parent(Mary,Donald)\}.
\end{align*}
\end{enumerate}
Consider now the query $NF(Steve,no)$ asking if $Steve$ is not diagnosed with NF. The answer to this query in the instance $I$ is {\bf true}. However, {\bf true} is not the consistent query answer  because of the repair $I_1$ (which indicates that the diagnosis of Steve may be incorrect).
\end{example}
We note that the framework of consistent query answers is parametrized by the notion of minimality used to define repairs. The original notion uses the symmetric set difference (between databases as sets of tuples) and set inclusion, i.e. the repairs are obtained by deleting and inserting a minimal set of tuples. This notion is most commonly considered in the literature and it is used in this paper. Other investigated notions of minimality use asymmetric set  difference~\cite{CaLeRo03} and the cardinality of the symmetric difference~\cite{ArBeCh03,LoBe07}. Finally, various notions of minimality have been considered to accommodate repairs obtained by attribute value modification~\cite{LoBe07,BBFL05,Lopatenko,BFFR05,Wij03}. 

It was observed very early that the number of possible repairs may be exponential even if we consider one functional dependency~\cite{ABCHRS03}. A na\"{i}ve approach to compute consistent query answers by materializing all repairs and consequently evaluating the query in every repair may thus be simply impractical. Consequently, to establish the tractability of database repairing and computing consistent query answers two fundamental decision problems have been investigated: (i) {\em repair checking} -- checking if a given database instance is a repair, and (ii) {\em consistent query answering} -- checking if an answer to a query is present in every repair. Most of the research in this area uses the notion of {\em data complexity}, commonly used to study tractability of computing answers in relational databases~\cite{Var82}. It allows to express the complexity of the problems in terms of the database size only: the set of integrity constraints and the query are assumed to be fixed. We note that the study of complexity of the two decision problems is motivated by the belief that tractable decision algorithms can be converted into efficient algorithms that compute consistent query answers and construct repair(s) of an inconsistent database. This belief is validated, for example, by existing polynomial-time algorithms where decision problems play a central role in the computation of consistent query answers~\cite{ChMaSt03,FuMi05}. 

The problems of repair checking and consistent query answering are parameterized by the class of integrity constraints. {\em Denial constraints} allow the user to specify sets of tuples that cannot be simultaneously present in the database because they create a conflict. This class of constraints includes {\em equality-generating dependencies}, thus also {\em functional dependencies}, and {\em exclusion constraints}~\cite{AbHuVi95}. The standard definition of denial constraints allows to use conjunctions of $=$ and $\neq$ comparisons to relate the values of the tuples creating conflicts. A more general version of denial constraints allows using any Boolean combination of formulas using $=$, $\neq$, $<$, $\leq$, $>$, and $\geq$~\cite{BaChWo99}. This version has also been studied in the context of consistent query answers~\cite{ChMa04}. There, the authors proposed the  {\em conflict hypergraph} to store all conflicts present in the database and subsequently use it to construct repairs and efficiently compute consistent query answers.

{\em Universal constraints} generalize denial constraints by allowing to express conflicts created not only by the presence of some tuples but also by simultaneous absence of other tuples. This class of constraints contains {\em full tuple-generating dependencies} (full TGDs) which have been thoroughly studied in the context of relational databases~\cite{AbHuVi95,Ma97,MaSr96}. Full TGDs contain an important class of {\em join dependencies} (JDs), and thus also its subclass {\em multi-valued dependencies} (MVDs), which are frequently used in the in the setting of {\em denormalized databases}~\cite{AbHuVi95,RaGe00}. In this context, the constraints are typically not actively enforced which permits the occurrence of {\em insertion/deletion/update anomalies}.

\begin{example} Consider a denormalized database that stores the information about locations and the offer of different chains of coffee shops. The schema is $\text{\it CoffeeShop}(Chain,Location,Beverage)$. The list of beverages offered by a particular chain is the same in every coffee shop, which is expressed with the following join dependency:
\[
\text{\it CoffeeShop}\jd{}[\{Chain,Location\},\{Chain,\text{\it Beverage}\}].
\]

Now, consider the instance in Figure~\ref{fig:coffee-instance}. 
\begin{figure}[htb]
\begin{center}
\begin{tabular}{|l|l|l|}
\multicolumn{3}{c}{\it CoffeeShop}\\
\hline
{\it Chain}& {\it Location} & {\it Beverage}\\
\hline
Starbucks & Delaware Ave. & Latte\\
\hline
Starbucks & Delaware Ave. & Espresso\\
\hline
Starbucks & Main Str. & Latte\\
\hline
Spot & Elmwood Ave. & Latte\\
\hline
\end{tabular}
\end{center}
\caption{\label{fig:coffee-instance} Inconsistent instance $I_2$.}
\end{figure}
This instance is inconsistent because Espresso is offered at Starbucks on Delaware Avenue but not at Starbucks on Main Street.
This instance has three repairs:
\begin{enumerate}
\item The first corresponds to the scenario where $Espresso$ has been added to the offer of $Starbucks$ but the change has not been propagated properly. This repair is obtained by {\em inserting} the tuple \begin{align*}
&\text{\it CoffeeShop}(Starbucks,Main~Str.,Espresso).
\intertext{
\item The second corresponds to the scenario where $Espresso$ has been removed from the offer of $Starbucks$ but the change has not been propagated properly. This repair is obtained by {\em deleting} the tuple
}
&\text{\it CoffeeShop}(Starbucks,Delaware~Ave., Espresso).
\intertext{
\item The third corresponds to the scenario where the coffee shop located on Main Street is being closed. This repair is obtained by {\em deleting} the tuple
}
&\text{\it CoffeeShop}(Starbucks, Main~Str., Latte).
\end{align*}
\end{enumerate}
Now, consider the query $\text{\it CoffeeShop}(Starbucks, Delaware~Ave., Latte)$. {\bf true} is the consistent answer to this query because it is {\bf true} in every repair. If we consider the query $\text{\it CoffeeShop}(Starbucks, Delaware~Ave.,Espresso)$, we note that the answer to this query is {\bf true} in the original answer. We observe, however, that {\bf true} is not the consistent query answer because the query is not {\bf true} in every repair. 
\end{example}
The complexity of repair checking and consistent query answering in the presence of general universal constraints has not been thoroughly studied. Previous research conducted in this area shows that computing consistent query answers is:
\begin{itemize}
\item in PTIME for the class of binary universal constraints and a restricted class of quantifier-free queries~\cite{ArBeCh99}.
\item in PTIME for the class of denial constraints and quantifier-free queries~\cite{ChMa04}.
\item in PTIME when at most one primary key per relation is present and the queries belong to a restricted class of conjunctive queries $C_{forest}$~\cite{FuMi05,Fuxman,FuMi07}.
\item coNP-complete for primary keys and arbitrary conjunctive queries~\cite{ChMa04}.
\item $\Pi_p^2$-complete for arbitrary sets of functional and inclusion dependencies~\cite{ChMa04}, when repairs are constructed using deletions only. 
\item undecidable for arbitrary sets of functional and inclusion dependencies~\cite{CaLeRo03}.
\end{itemize}
We remark that the class of universal constraints captures only {\em full inclusion dependencies} and in the paper we do not consider general inclusion dependencies. 

In this paper we investigate computing consistent query answers and repairs in the presence of {\em universal constraints}. This research constitutes a continuation and substantial extension of~\cite{ChMa04}.
Similarly to~\cite{ChMa04} in the constraint definition we allow any Boolean combination of atomic formulas that use the binary relations: $=$, $\neq$, $<$, $\leq$, $>$, and $\geq$. We propose an {\em extended} version of the conflict hypergraph whose hyperedges span both tuples present and absent in the database. The size of the extended conflict hypergraph is still polynomial in the size of the database and every repair corresponds to a maximal independent set. Although, the converse correspondence is not necessarily true, i.e.,  not every maximal independent set defines a repair, we consider the extended conflict hypergraph to be a {\em compact representation} of all repairs.

Next, we study the computational implications of universal constraints. In this paper we show that:
\begin{itemize}
\item The complexity of repair checking is:
\begin{itemize}
\item in PTIME for arbitrary full tuple-generating dependencies and denial constraints. Consequently, we present a polynomial {\em database repairing} algorithm that is both {\em sound} (always constructing a repair) and {\em complete} (able to construct every repair). 
\item coNP-complete for arbitrary universal constraints.
\end{itemize}
\item The complexity of consistent query answering is:
\begin{itemize}
\item in PTIME for quantifier-free closed queries in the presence of join dependencies, {\em acyclic} full tuple-generating dependencies, and denial constraints.
\item coNP-complete for \rrmark{Specified the precise query classes.}\rrtext{atomic queries} in the presence of arbitrary full tuple-generating dependencies and denial constraints. 
\item $\Pi_p^2$-complete for \rrtext{atomic queries} in the presence of  arbitrary universal constraints. 
\end{itemize}
\end{itemize}

The paper is organized as follows. Section~\ref{sec:basic} contains basic notions and definitions. In Section~\ref{sec:extended} we present the extended conflict hypergraph, study its properties, and investigate basic properties of the framework. In Section~\ref{sec:ftgd-repairs} we study the complexity of repair checking in the presence of full tuple-generating dependencies and we present a database repairing algorithm. In Section~\ref{sec:ftgd-cqa} we investigate the complexity of consistent query answering in the presence of full tuple-generating dependencies. In Section~\ref{sec:univ} we investigate the complexity of consistent query answering and repair checking in the presence of arbitrary universal constraints. In Section~\ref{sec:related} we discuss related work. Section~\ref{sec:future} contains final conclusions and the discussion of future work. 

\section{Preliminaries}
\label{sec:basic}
A database {\em schema} $\mathcal{S}$ is a set of relation names of fixed arity (greater than $0$) and we use $R,P,\ldots$ to denote relation names. Relation attributes are drawn from an infinite set of names $U$, and we use $A,B,C,\ldots$ to denote elements of $U$ and $X,Y,Z,\ldots$ to denote finite subsets of $U$. For $R\in\mathcal{S}$ we denote the set of all attributes of $R$ by $attrs(R)$. Every element of $U$ is typed and we consider only two disjoint infinite domains: $\mathsf{Q}$ (rationals) and $D$ (uninterpreted constants). We assume that two constants are equal if and only if they have the same name, and we allow the standard {\em built-in} relation symbols $=$ and $\neq$ over $D$. We also allow the built-in relation symbols $=$, $\neq$, $<$, $\leq$, $>$, and $\geq$ with their natural interpretation over $\mathsf{Q}$. 
We use these symbols together with the vocabulary of relational names $\mathcal{S}$ to build a first-order language $\mathcal{L}$. An $\mathcal{L}$-formula is:
\begin{itemize}
\item {\em closed} (or a {\em sentence}) if it has no free variables,
\item {\em ground} if it has no variables whatsoever,
\item {\em quantifier-free} if it has no quantifiers,
\item {\em atomic} if it has no quantifiers and no Boolean connectives.
\end{itemize}
Finally, a {\em fact} is an atomic ground $\mathcal{L}$-formula and a {\em literal} is a fact or the negation of a fact.

Database {\em instances} are finite, first-order structures over the schema. Often, we will find it more convenient to view an instance $I$ as the finite set of all facts satisfied by the instance $\{R(t)\sepbar R\in \mathcal{S}, I\models R(t)\}$.

In the sequel, we will denote tuples of variables by $\bar{x},\bar{y},\ldots$, tuples of constants by $t,s,\ldots$, facts by $p,q,r,\ldots$, quantifier-free formulas using only built-in predicates by $\varphi$, and instances by $I,J,\ldots$ 

\subsection{Integrity constraints}
An integrity constraint is an $\mathcal{L}$-sentence, i.e. a closed first-order $\mathcal{L}$-formula. In this paper we consider the class of {\em universal constraints}, $\mathcal{L}$-sentences of the form
\begin{equation}\label{eq:univ}
\forall \bar{x}\sepdot \neg [
R_1(\bar{x}_1)\land \ldots \land R_n(\bar{x}_n) \land
\neg P_1(\bar{y}_1)\land \ldots \land \neg P_m(\bar{y}_m) \land \varphi(\bar{x})
],
\end{equation}
where $\varphi(\bar{x})$ is a quantifier-free formula referring to built-in relation names only and 
$\bar{y}_1\cup\ldots\cup\bar{y}_m\subseteq
\bar{x}_1\cup\ldots\cup\bar{x}_n=\bar{x}$ (this is a standard safety requirement~\cite{AbHuVi95}). Also, we make a natural assumption that $n+m>0$. The constraint \eqref{eq:univ} will be often presented as:
\begin{equation}
\label{eq:univ-impl}
R_1(\bar{x}_1)\land\ldots\land R_n(\bar{x}_n)\land \varphi(\bar{x})
\rightarrow
P_1(\bar{y}_1)\lor\ldots\lor P_m(\bar{y}_m),
\end{equation}
where all the variables are implicitly universally quantified. 

The class of universal constraints contains the following basic classes of integrity constraints:
\begin{enumerate}
\item {\em Full tuple-generating dependencies (full TGDs)}: universal constraints with one atom in rhs ($m=1$). \rrtext{Often, the full TGDs considered in literature have a conjuntion of atoms in rhs. 
\rrmark{Comment on single- and multi-head full TGDs.}
We remark, however, that a {\em multi-head} full TGD is equivalent to a set of {\em single-head} full TGDs~\cite{AbHuVi95}.}
\item {\em Join dependencies (JDs)} commonly formulated as 
$R\jd{}[X_1,\ldots,X_k]$, where $R$ is a relation name and $X_1,\ldots,X_k$ are subsets of attributes of $R$ whose union contains all attributes of $R$. A common relational algebra definition is: $R=\pi_{X_1}(R)\Join\ldots\Join\pi_{X_k}(R)$. The equivalent full tuple-generating dependency is:
\[
R(\bar{x}_1)\land\ldots\land R(\bar{x}_k) \land
\bigwedge_{1\leq i,j \leq k} 
\bar{x}_i[X_i\cap X_j]=\bar{x}_j[X_j\cap X_i]
\rightarrow R(\bar{y}),
\]
where $\bar{z}[Y]$ is the subvector of $\bar{z}$ that corresponds to the attributes in $Y$, and $\bar{y}\subseteq \bar{x}_1\cup\ldots\cup\bar{x}_n$ such that $\bar{y}[X_i\setminus\bigcup_{1\leq j < i} X_j] = 
\bar{x}_i[X_i\setminus\bigcup_{1\leq j < i} X_j]$ for $i\in\{1,\ldots,k\}$. 
\item {\em Denial constraints}: universal constraints with no atoms in the rhs ($m=0$):
\[
R_1(\bar{x}_1)\land\ldots\land R_n(\bar{x}_n)\land\varphi(\bar{x})\rightarrow \mathbf{false}.
\]
\item {\em Functional dependencies (FDs)} commonly formulated as $R\fd{}X\rightarrow Y$, where $X$ and $Y$ are sets of attributes of $R$. An FD $R\fd{}X\rightarrow Y$ is expressed by the following denial constraint:
\[
R(\bar{x}_1)\land R(\bar{x}_2)\land 
\bar{x}_1[X]=\bar{x}_2[X]\land
\neg(\bar{x}_1[Y]=\bar{x}_2[Y])\rightarrow \mathbf{false}.
\]
\end{enumerate}

The following restriction will allow us to identify a tractable class of integrity constraints. 

\begin{definition}[Acyclic constraints]
The {\em dependency graph} $\mathcal{D}(\mathcal{S},F)$ of a set of universal
constraints $F$ is a directed graph whose set of vertices is the relational schema $\mathcal{S}$ and for any constraint \eqref{eq:univ} in $F$ there is an edge from $P_j$ to $R_i$ for every $i\in\{1,\ldots,n\}$ and every $j\in\{1,\ldots,m\}$. The set of constraints $F$ is {\em acyclic} if $\mathcal{D}(\mathcal{S},F)$ is an acyclic graph.  
\end{definition}
We also adapt the standard notions of {\em height} and {\em depth} of a node in a tree to (possibly cyclic) dependency graphs~\cite{Cormen}. Given a dependency graph $\mathcal{D}(\mathcal{S},F)$, the {\em acyclic depth} of a node $R$, denoted $depth(R)$, is the maximal length of a directed acyclic path that ends in $R$, where the length of a path is the number of edges its comprises of. The {\em acyclic height} of $R$, denoted $height(R)$, is the maximal length of a directed acyclic path that originates in $R$. The {\em acyclic height} of $\mathcal{D}(\mathcal{S},F)$ is maximum acyclic height of all node in $\mathcal{D}(\mathcal{S},F)$. We note that both the acyclic height and acyclic depth of a node are bounded by the acyclic height of the dependency graph. 
\begin{example}
Figure~\ref{fig:dependency-graph} presents a dependency graph for schema 
\[
\mathcal{S}=\{R(A,B),P(A,B),T(A,B,C),S(A,B,C)\}
\] 
and a cyclic set of constraints
\begin{align*}
F=\{
&R(x,y)\land P(y,z) \rightarrow S(x,y,z),
 S(x,y,z)\rightarrow T(x,y,z),\\
&T(x,y,z)\rightarrow P(x,y)\lor P(y,z)
\}.
\end{align*}
\begin{figure}[htb]
\begin{center}
\begin{tikzpicture}[scale=1]
\node (R) at (0,0) {$R$};
\node (P) at (0,1) {$P$};
\node (S) at (1,0) {$S$};
\node (T) at (1,1) {$T$};

\draw[->] (S) -- (R);
\draw[->] (S) -- (P);
\draw[->] (P) -- (T);
\draw[->] (T) -- (S);

\node at (4,0.5) {
\small
\begin{tabular}{r|l|l|l|l}
& $R$ & $P$ & $S$ & $T$ \\
\hline
$depth$ & $3$ & $2$ & $2$ & $2$ \\
\hline
$height$ & $0$ & $3$ & $2$ & $2$  
\end{tabular}
};
\end{tikzpicture}
\end{center}
\caption{\label{fig:dependency-graph} A cyclic dependency graph.}
\end{figure}
The acyclic height of this graph is $3$. 
\end{example}
Database consistency is defined in the standard way.
\begin{definition}
Given a database instance $I$ and a set of integrity constraints $F$, $I$ is {\em consistent} with $F$ if $I\models F$ in the standard model-theoretic sense; otherwise $I$ is {\em inconsistent}.
\end{definition}
We observe that because we do not allow relation names of arity $0$ and we consider universal constraints satisfying the safety requirement, the constraint~\eqref{eq:univ} can have negative atoms only if it has positive atoms as well, i.e. $m>0$ implies $n>0$. Therefore, the prerequisite of a constraint violation is the existence of some facts in the database. Consequently, if the instance is empty then all the constraints are satisfied. We note this conforms to the behavior of typical SQL database management systems: an empty database satisfies any set of constraints expressed in SQL.

\subsection{Queries}
In this paper we deal only with {\em closed} queries, i.e. closed $\mathcal{L}$-formulas. The query answers are Boolean: {\bf true} or {\bf false} A query is {\em atomic} ({\em quantifier-free}) if the $\mathcal{L}$-formula is atomic (quantifier-free respectively). A {\em conjunctive} query is an existentially quantified conjunction of atomic $\mathcal{L}$-formulas.
\begin{definition}
{\bf true} is the answer to a closed query $Q$ in an instance $I$ if $I\models Q$; otherwise the answer to $Q$ is {\bf false}.
\end{definition}

\subsection{Repairs}
Repairs are defined as consistent instances that are minimally different from the original one. The differences are measured in terms of the set of facts that need to be deleted and inserted. Because we view an instance as the set of facts, the {\em (symmetric) difference} $\Delta(I_1,I_2)$ of the instances $I_1$ and $I_2$ is defined as $\Delta(I_1,I_2)=I_1\setminus I_2 \cup I_2 \setminus I_1$. 

Given an instance $I$, the {\em relative proximity relation} $\leq_I$ on instances is defined as 
\[
I_1\leq_I I_2 \iff \Delta(I,I_1) \subseteq \Delta(I,I_2).
\]
We note that  $\leq_I$ is a partial order and we write $I_1 <_I I_2$ if $I_1\leq_I I_2$ and $I_1\neq I_2$. 
\begin{definition}[\cite{ArBeCh99}]\label{def:repair}
Given a set of integrity constraints $F$ and database instances $I$ and $I'$, we say that $I'$ is a {\em repair} of $I$ w.r.t. $F$ if $I'$ is a $\leq_I$-minimal instance consistent with $F$. By $Repairs(I,F)$ we denote the set of all repairs of $I$ w.r.t. $F$. 
\end{definition}
Because an empty instance over a schema always satisfies any set of universal constraints over the schema, the set of repairs is guaranteed to be non-empty. 

\subsection{Consistent query answers}
Finally, we use the repairs to define the consistent query answers.
\begin{definition}[\cite{ArBeCh99}] \label{def:cqa}
{\bf true} {\rm (}{\bf false}{\rm)} is the {\em consistent answer} to a closed query $Q$, denoted $I\models_F Q$ {\rm(}$I\not\models_F Q$ resp.{\rm)} if and only if {\bf true} {\rm(}{\bf false} resp.{\rm)} is the answer to $Q$ in every repair of $I$ w.r.t. $F$.
\end{definition}
We note that our approach can be easily extended to open queries along the lines of~\cite{ChMa04,ChMaSt04}. In essence, from an open query $Q(\bar{x})$ we derive a query $Q^E(\bar{x})$ that defines an {\em envelope}, a superset of consistent query answers to $Q(\bar{x})$. For every tuple $t$ from the envelope, $t$ is a consistent answer to $Q(\bar{x})$ if and only if {\bf true} is the consistent answer to the closed query $Q(t)$.

\subsection{Decision problems}
We consider here the following complexity classes:
\begin{itemize}
\item {\em PTIME}: the class of decision problems solvable in polynomial time by deterministic Turing machines;
\item {\em coNP}: the class of decision problems whose complements are solvable in polynomial time by nondeterministic Turing machines;
\item $\Pi^p_2$: the class of decision problems whose complements are solvable in polynomial time by nondeterministic Turing machines with an NP oracle. 
\end{itemize}

To investigate tractability of the framework of consistent query answers we use the notion of data complexity~\cite{Var82}. This notion allows to describe the complexity of the problems in terms of the size of the database only and assume the remaining parts of the input to be fixed. There are two classical decision problems that are investigated in the context of consistent query answers~\cite{ChMa04}:  
\begin{enumerate}
\item[$(i)$] {\em repair checking} -- determining if a database
  instance is a repair of a given database instance, i.e. the
  complexity of the following set
\[
\mathcal{B}_F =\{(I,I') \sepbar I' \in Repairs(I,F)\}.
\]
\item[$(ii)$] {\em consistent query answering} -- determining if 
{\bf true} is the consistent answer to a given closed query in a
given database  w.r.t. a given set of integrity constraints,
i.e. the complexity of the following set
\[
\mathcal{D}_{F,Q} = \{I \sepbar I \models_F Q\}.
\]
\end{enumerate}

\section{Basic constructions and facts}
\label{sec:extended}
In this section we generalize the conflict hypergraph for denial constraints~\cite{ABCHRS03,ChMa04}. In the scope of this section, we fix an instance $I$ and a set of universal constraints $F$. 

\subsection{Extended conflict hypergraph}

For denial constraints, a conflict is a set of facts whose {\em presence} violates a constraint. For universal constraints, a conflict is created not only by the {\em presence} of some facts but also by the simultaneous {\em absence} of other facts. 
\begin{definition}[Conflict]
A set of literals
\[
\{R_1(t_1),\ldots,R_n(t_n),\neg P_1(s_1),\ldots,\neg P_m(s_m)\}
\] 
is a {\em conflict} w.r.t. a constraint
\begin{equation*}
\forall \bar{x}\sepdot \neg [
R_1(\bar{x}_1)\land \ldots \land R_n(\bar{x}_n) \land
\neg P_1(\bar{y}_1)\land \ldots \land \neg P_m(\bar{y}_m) \land \varphi(\bar{x})
],
\end{equation*}
if there exists a ground substitution $\theta$ of variables $\bar{x}$ such that $\theta(\bar{x}_i)=t_i$ for $i\in\{1,\ldots,n\}$, $\theta(\bar{y}_j)=s_j$ for $j\in\{1,\ldots,m\}$, and $\varphi(\theta(\bar{x}))$ holds.
\end{definition}

If the constraints are limited to denial constraints, conflicts can be resolved only by deleting facts. Moreover, deleting a fact will not create further conflicts. Therefore, only conflicts created by the facts from the original instance need to be considered. 

In the case of universal constraints, the picture is more complex. First, a conflict can be resolved not only by deleting facts but also by inserting a fact. Second, deleting a fact can create conflicts caused by the absence of the fact. Similarly, inserting a fact can create conflicts caused by the presence of the fact. Moreover, a cascading propagation of conflicts can easily take place: resolving one conflict leads to the creation of another conflict whose resolution leads to yet another one, and so on. 

\begin{example}\label{ex:cascade}
Consider a schema $\mathcal{S}=\{R(A,B),P(C)\}$ with one constraint $F=\{R(x,y)\land P(x)\rightarrow P(y)\}$ and take $I=\{R(1,2),R(2,3),P(1)\}$. The conflict $\{R(1,2),P(1),\neg P(2)\}$ can be resolved by inserting $P(2)$ into the instance $I$ to obtain: $I\cup\{P(2)\}$. This creates the conflict $\{R(2,3),P(2),\neg P(3)\}$ which can be resolved by further inserting $P(3)$. Finally, we obtain the repair $I_1=I\cup\{P(2),P(3)\}$. The other repairs are $I_2=I\setminus\{R(1,2)\}$, $I_3=I\setminus\{P(1)\}$, and $I_4=I\cup\{P(2)\}\setminus\{R(2,3)\}$.
\end{example}

Clearly, to capture all repairs it is not enough to consider only the facts present in the original database instance. Also the facts potentially inserted need to be considered. We capture the set of relevant facts in the following way.

\begin{definition}[Hull]\label{def:hull}
The {\em hull} is the minimal set of literals $Hull(I,F)$ satisfying the following conditions:
\begin{enumerate}
\item $I\subseteq Hull(I,F)$,
\item if a set $e$ is a conflict w.r.t. a constraint in $F$ such that every fact of $e$ belongs to $Hull(I,F)$, then for every $\neg P(t)$ in $e$, both $\neg P(t)$ and $P(t)$ belong to $Hull(I,F)$. 
\end{enumerate}
\end{definition}

We note that this definition can be easily translated to a negation-free Datalog program which computes the set of literals in the hull. The arities of the predicates are equal to the arities of the corresponding relation names.
\begin{example}\label{ex:datalog1}
For the set of constraints $F$ from Example~\ref{ex:cascade} we construct the following Datalog program:
\begin{align*}
R^{H}(x,y) &\leftarrow R(x,y).\\
P^{H}(x) &\leftarrow P(x).\\
P^{H}(y) &\leftarrow R^{H}(x,y), P^{H}(x).\\
\bar{P}^{H}(y) &\leftarrow R^{H}(x,y), P^{H}(x).
\end{align*}
Now, if we treat the instance $I$ as the extensional database, the program above has the following solution (least fixpoint):
\begin{align*}
I\cup&\{
R^{H}(1,2),R^{H}(2,3),P^{H}(1)\}\cup\{P^{H}(2),P^{H}(3),\bar{P}^{H}(2),\bar{P}^{H}(3)\},
\end{align*}
which corresponds to $Hull(I,F)=I\cup\{P(2),\neg P(2), 
P(3), \neg P(3)\}$.
\end{example}

Our intention is to use hyperedges to restrict the sets of literals used to construct a repair so that no conflicts are present. Because the hull may contain a fact and its negation, we also use an edge to prevent us from considering these two together. 
\begin{definition}[Extended conflict hypergraph]
The {\em extended conflict hypergraph} $G(I,F)$ is a hypergraph whose set of vertices is $Hull(I,F)$ and whose set of hyperedges consists of the following two types of sets: 
\begin{enumerate}
\item \underline{{\em conflict hyperedges}} $e\subseteq Hull(I,F)$ such that $e$  is a conflict w.r.t a constraint in $F$,
\item \underline{{\em stabilizing edges}} $\{P(t),\neg P(t)\}$ such that both $P(t)$ and $\neg P(t)$ belong to $Hull(I,F)$.
\end{enumerate}
An {\em independent set} of the extended conflict hypergraph $G(I,F)$
is any subset of $Hull(I,F)$ that contains no hyperedges. $M$ is a {\em maximal} independent set if there exists no independent set $M'\subseteq Hull(I,F)$ such that $M\subsetneq M'$. 
\end{definition}

\begin{example}
Figure~\ref{fig:ech} contains the extended conflict hypergraph $G_1$ for the instance $I$ w.r.t the set of universal constraints $F$ from Example~\ref{ex:cascade}.
\begin{figure}[htb]
\begin{center}
\begin{tikzpicture}[yscale=1.25,xscale=1.75]
\node (r1) at (0,0) {$R(1,2)$};
\node[inner sep=0pt] (p1) at (-.5,-1) {$P(1)$};
\node[inner sep=0pt] (np2) at (.5,-1) {$\neg P(2)$};
\draw[blue!50!black] (0,0) plot[smooth cycle, tension=0.85] coordinates{
+(0,.4) +(.85,-1.15) +(-.85,-1.15)
};

\node (r2) at (2,0) {$R(2,3)$};
\node[inner sep=0pt] (p2) at (1.5,-1) {$P(2)$};
\node[inner sep=0pt] (np3) at (2.5,-1) {$\neg P(3)$};
\draw[blue!50!black] (2,0) plot[smooth cycle, tension=0.85] coordinates{
+(0,.4) +(.85,-1.15) +(-.85,-1.15)
};

\node[inner sep=0pt] (p3) at (3.5,-1) {$P(3)$};

\draw[densely dotted] (np2) -- (p2);
\draw[densely dotted] (np3) -- (p3);

\end{tikzpicture}
\end{center}
\caption{\label{fig:ech} The extended conflict hypergraph $G_1$ for $I$ and $F$ from Example~\ref{ex:cascade}.}
\end{figure}
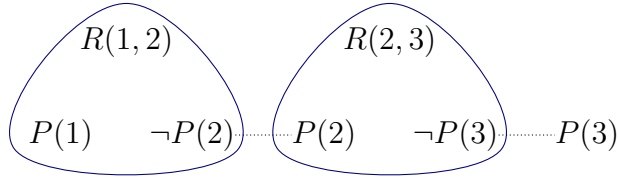
A dotted line is used for stabilizing edges connecting a fact and its negation (if present). We observe that deleting $P(1)$ does not lead to a conflict, and consequently, $\neg P(1)$ is not present in the hull. 
\end{example}
Since we assume the set of constraints to be fixed, the cardinalities of each conflict in $G(I,F)$ are bounded by the maximal number of atoms used in a constraint definition, a constant $K$. To construct the set of hyperedges we need to consider all subsets of cardinality bounded by $K$. Consequently, 
\begin{proposition} 
The extended conflict hypergraph $G(I,F)$ can be constructed in time polynomial in the size of $I$ (data complexity). Also, the size of $G(I,F)$ is polynomial in the size of $I$. 
\end{proposition}

The presented extension of the conflict hypergraph is backward-com\-pat\-i\-ble with~\cite{ChMa04}: if we restrict the set of constraints to denial constrains only, the hull is equal to the original instance and we obtain the standard conflict hypergraph~\cite{ABCHRS03,ChMa04}. Finally, in the presence of denial constraints, because the repairs are obtained by deleting facts only, the repairs are maximal consistent subsets of the original instance. 
\begin{proposition}[\cite{ChMa04}]\label{prop:denial-repairs}
If $F$ is a set of denial constraints, each repair of $I$ w.r.t. $F$ corresponds to a maximal independent set of $G(I,F)$, and vice versa. 
\end{proposition}

\rrmark{Added example. Changed ``{\em saturation}'' to a more intuitive ``{\em complementation}''.}\rrtext{The same equivalence does not hold for the extended conflict hypergraph. For instance, while the repair $I_1$ (Example~\ref{ex:cascade}) is a maximal independent set of $G_1$ (Fig.~\ref{fig:ech}), the repairs $I_2$, $I_3$, and $I_4$ are not. 
One reason for this is the use of negated facts in the extended conflict hypergraph. We observe, however, that if we complement $I_2=\{R(2,3),P(1)\}$ with the (relevant) negations of facts that are not present in $I_2$, we obtain a maximal independent set $\{R(2,3),P(1),\neg   P(2),\neg P(3)\}$. This holds for every repair.} 
\begin{proposition}\label{prop:saturated}
For any repair $I'\in Repairs(I,F)$ the set 
\[
Compl(I')=I'\cup\{\neg R(t)\in Hull(I,F) \sepbar R(t)\not\in I'\}
\]
is a maximal independent set of $G(I,F)$. 
\end{proposition}
\begin{proof}
Naturally, $Compl(I')$ is independent because $I'$ is consistent. Before showing that $Compl(I')$ is a maximal independent set we make two observations following from the construction of the hull:
\begin{enumerate}
\item[$1^o$]  For every $\neg R(t)$ in $Hull(I,F)$, the fact $R(t)$ is also present in $Hull(I,F)$. 
\item[$2^o$] For every fact $R(t)$, if $\neg R(t)$ is not present in $Hull(I,F)$, then $R(t)\in I$.
\end{enumerate}
Now, take any fact $p\in Hull(I,F)\setminus Compl(I')$. If $p=\neg R(t)$, then $R(t)\in Compl(I')$ and adding $\neg R(t)$ to $Compl(I')$ introduces the stabilizing edge $\{R(t),\neg R(t)\}$. Analogously, we show that $Compl(I')$ cannot be extended with a fact $R(t)$ whose negation is present in $Hull(I,F)$. The only remaining case is extending $Compl(I')$ with a fact $R(t)$ that belongs to $I$. We observe that for $R(t)\in I\setminus I'$, we have $I'\cup\{R(t)\}<_II'$ and so $I'\cup\{R(t)\}$ is inconsistent. Consequently, $Compl(I')\cup\{R(t)\}$ contains a conflicting hyperedge. 
\qed
\end{proof}

We note that the converse of Proposition~\ref{prop:saturated} is not necessarily true, i.e. for a maximal independent set $M$ of the extended conflict hypergraph, its {\em positive projection} $M^+=\{R(t)\sepbar R(t)\in M\}$ needs not to be a repair. For instance, if we take the maximal independent set $M=\{R(1,2),\neg P(2),R(2,3),\neg P(3)\}$ of $G_1$ (Fig.~\ref{fig:ech}), its positive projection $M^+=\{R(1,2),R(2,3)\}$ is not a repair. Nevertheless, the extended conflict hypergraph allows us to capture all the repairs. And since its size is polynomial in the size of $I$, we consider it to be a {\em compact representation} of the repairs of $I$. 
\begin{proposition}\label{prop:repair-mis}
For any maximal independent set $M$ of $G(I,F)$ either $M^+$ is a repair of $I$ w.r.t. $F$ or there exists a maximal independent set $N$ of $G(I,F)$ such that $N^+<_I M^+$.
\end{proposition}
\begin{proof}
Take any maximal independent set $M$ of $G(I,F)$ such that $M^+$ is not a repair. Naturally, $M^+$ is consistent and therefore there exists a repair $I'$ such that $I'<_I M^+$. It suffices to note that $Compl(I')^+=I'$ and by Proposition~\ref{prop:saturated} $Compl(I')$ is a maximal independent set of $G(I,F)$. 
\qed
\end{proof}
\subsection{Grounding constraints}
Often, we will find it more convenient to view conflict hyperedges as grounded integrity constraints. This helps to pinpoint the exact reasons for integrity violations and the facts that can be inserted and deleted to resolve the conflict. 
\begin{definition}
For any conflict hyperedge $$\{R_1(t_1),\ldots,R_n(t_n),\neg P_1(s_1),\ldots,\neg P_m(t_m)\}$$ in $G(I,F)$ the implication $$R_1(t_1)\land\ldots\land R_n(t_n)\rightarrow P_1(s_1)\lor\ldots\lor P_m(s_m)$$ is a {\em ground rule} (or simply {\em rule}) in $G(I,F)$. By $Rules(I,F)$ we denote the set of all ground rules in $G(I,F)$.
\end{definition}
A denial (full TGD, JD, or non-JD, resp.) rule is a rule obtained from a conflict w.r.t. a denial constraint (full TGD, JD, or non-JD resp.)
The facts in the lhs and the rhs of a ground rule are represented as sets, i.e. no particular order is assumed and duplicates are removed. 

Naturally, the cardinality of the set of the ground rules is equal to the number of the conflict hyperedges, and thus it is polynomial in the size of $I$. We also note that when considering the instances using facts from the hull only, satisfaction of the set of constraints $F$ implies the satisfaction of $Rules(I,F)$. The converse is also true because the hull contains all relevant facts. 
\begin{proposition} For any instance $J$ such that $J\subseteq Hull(I,F)$, $J\models F$ if and only if $J\models Rules(I,F)$. 
\end{proposition}

\section{Repairing in the presence of full tuple-generating dependencies}
\label{sec:ftgd-repairs}
Now, we show that repair checking in the presence of full tuple-generating dependencies and denial constraints is tractable. We use the result to construct a {\em complete} and {\em sound} repairing algorithm. In the scope of this section we fix an instance $I$ and a set $F$ of denial constraints and full tuple-generating dependencies.

\subsection{Repair checking}
\rrmark{Sections 3.3 and 4.1 merged.} 
We begin by presenting an alternative characterization of repairs w.r.t. a set of full TGD and denial constraints. 
If we view full TGDs as Datalog programs we can use the standard consequence operator to identify the facts that need to be added to satisfy the constraints. 
\begin{definition}
For a set of facts $J\subseteq Hull(I,F)$, the operator of {\em immediate consequence} of $F$ on $J$ is defined as
\begin{align*}   
  T_F(J)=J\cup \{ P(s) \sepbar 
  & R_1(t_1)\land\ldots\land R(t_n)\rightarrow P(s)\in Rules(I,F)
    \quad{s.t.}\\
  & \{R(t_1),\ldots,R(t_n)\}\subseteq J\}.
\end{align*}
$T_F^*$ is defined as the transitive closure of $T_F$.  
\end{definition}
It is a classical result that $T_F^*(J)$ can be computed in time polynomial in the size of $J$~\cite{Baral03}. Now, we present an alternative characterization of a repair w.r.t. a set of full TGDs and denial constraints.
\rrtext{Essentially, every repair is obtained by closing under $T_F^*$ some subset of the original instance and verifying that the resulting instance is consistent.}
\begin{lemma}\label{lemma:repair-test}
$I'$ is a repair of $I$ w.r.t. $F$ if and only if the following conditions are satisfied:
\begin{enumerate}
\item[$(i)$] $I'$ is consistent,
\item[$(ii)$] $T^*_F(I'\cap I)=I'$,
\item[$(iii)$] \rrtext{\rrmark{More intuitive presentation}
there is no $R(t)\in I\setminus I'$ such that $J'=T^*_F(I'\cup\{R(t)\})$ is consistent and $J'\setminus I = I'\setminus I$.
}
\end{enumerate}
\end{lemma}
\begin{proof}
For the {\em only if\/} part $I'\in Repairs(I,F)$ implies that (i) holds. 

To show (ii) we note that $T^*_F(I\cap I')$ by definition is the minimal set that contains $I\cap I'$ and satisfies all full TGDs from $F$. Hence, $T^*_F(I\cap I')\subseteq I'$ and, as a subset of a consistent instance $I'$, $T_F^*(I\cap I')$ satisfies also all denial constraints from $F$. We also note that $T_F^*(I\cap I')$ and $I'$ agree on the facts in $I$.  $I'$ is the $\leq_I$-minimal consistent instance that contains all of $I\cap I'$ and none of $I\setminus I'$, which implies that $I'=T^*_F(I\cap I')$.

To show (iii) we take any $R(t)\in I\setminus I'$ such that $J'=T_F^*(I'\cup \{R(t)\})$ is consistent. Since $I'$ is a $\leq_I$-minimal consistent instance, $J'\not\leq_I I'$. This implies  $J'\setminus I\not\subseteq I'\setminus I$ because $I\setminus J'\subseteq I\setminus I'$. 

For the {\em if\/} part take any $\leq_I$-minimal consistent instance $I''$ such that $I''\leq_I I'$. Such instance exists because by (i) $I'$ is consistent. $I''$ is a repair, and it satisfies (ii). Therefore it suffices to show that $I'\cap I = I''\cap I$. $I''\leq_I I'$ shows directly that $I'\cap I \subseteq I'' \cap I$. This also shows that $I'\subseteq I''$. 

To show $I''\cap I\subseteq I'\cap I$, suppose there exists $R(t)\in I''\cap I$ such that $R(t)\not\in I'\cap I$. Naturally, $T^*_F(I'\cup\{R(t)\})\subseteq I''$ and since $I'\subseteq I''$, we have \rrmark{Changed formatting.}
\[
T^*_F(I'\cup\{R(t)\})\setminus I \subseteq I''\setminus I\subseteq I'\setminus I.
\]
On the other hand, we note that $R(t)\in I\setminus I'$ and $T^*_F(I'\cup\{R(t)\})$ is consistent as a subset of a consistent instance closed under $T^*_F$. By (iii), we obtain
\[
T^*_F(I'\cup\{R(t)\}) \not\subseteq I'\setminus I,
\]
which is a contradiction. Thus, $I'=I''$ and $I'$ is a repair. 
\qed
\end{proof}
We observe that the conditions of Lemma~\ref{lemma:repair-test} can be checked in time polynomial in the size of the instances $I$ and $I'$. Consequently, 
\begin{theorem} \label{thm:rcheck-cyclic}
Repair checking is in PTIME for any set of denial constraints and 
full tuple-generating dependencies. 
\end{theorem}

\subsection{Constructing a repair}
In this subsection we present a polynomial-time algorithm for constructing repairs of an instance w.r.t. a set of full TGDs and denial constraints. Rather than trying to resolve all the conflicts present in the instance, the algorithm constructs a repair from scratch: it begins with an empty instance, iterates over the facts of the original instance, and for every fact makes a decision whether to {\em discard} the fact or to {\em add} it to the constructed instance. It should be noted that those two actions, although related, are different and should not be confused with {\em inserting} and {\em deleting} facts from the original instance in order to resolve conflicts. 

Before we present the algorithm for full TGDs and denial constraints, we recall a simpler algorithm~\cite{StChMa06} that constructs repairs in the presence of denial constraints. 
\begin{algorithm}[htb]
\caption{\label{alg:sound-repair} Constructing a repair of $I$ w.r.t. a set of denial constraints $F$}
\begin{tabbing}
mm\=xxx\=xxx\=xxx\=xxx\=xxx\=\kill
{\footnotesize\tt ~1:} \> $I^o\gets I$\\
{\footnotesize\tt ~2:} \> $J\gets\varnothing$\\
{\footnotesize\tt ~3:} \> {\bf while} $I^o\neq\varnothing$ {\bf do}\\
{\footnotesize\tt ~4:} \>\> {\bf choose} $R(t)\in I^o$\\
{\footnotesize\tt ~5:} \>\> $I^o\gets I^o\setminus\{R(t)\}$\\
{\footnotesize\tt ~6:} \>\> {\bf if} $J\cup\{R(t)\}\models F$ {\bf then}\\
{\footnotesize\tt ~7:} \>\>\> $J\gets J\cup\{R(t)\}$\\
{\footnotesize\tt ~8:} \> {\bf return} $J$
\end{tabbing}
\end{algorithm}
Algorithm~\ref{alg:sound-repair} constructs a maximal consistent subset of the input instance by iterating over the facts of the input instance (using $I^o$ to store the remaining facts) and adding the current fact if doing so does not violate the constraints. Since maximal consistent subsets of the original instance correspond to maximal independent set of the conflict hypergraph, by Proposition~\ref{prop:denial-repairs} this algorithm always returns a repair, i.e. it is {\em sound}. We also note that it is {\em complete}, i.e. every repair of the original instance can be constructed: it suffices to chose first the facts of the desired repair. 

An approach that constructs a maximal consistent subset of the original instance is also sound for constructing a repair in the presence of denial constraints and general TGDs~\cite{Ko08}. However, in the presence of TGDs a repair needs not be a subset of the original instance and consequently this approach is not necessarily complete. 
\begin{example}
For the schema $\mathcal{S}=\{R(A),P(A)\}$ with a set of constraints $F=\{R(x)\rightarrow P(x)\}$ consider the instance $I=\{R(1),R(2)\}$. This instance has four repairs w.r.t. $F$:
\begin{align*}
I_1'=\varnothing, I_2'=\{R(1),P(1)\}, I_3'=\{R(2),P(2)\}, I_4'=\{R(1),P(1),R(2),P(2)\}.
\end{align*}
Only the repair $I_1'$ is a subset of the original instance. 
\end{example}

Our approach extends Algorithm~\ref{alg:sound-repair} by allowing it also to add a fact together with the facts implied by full TGDs. In this way, for instance, the repair $I_2$ is obtained by adding $R(1)$ together with $P(1)$ and discarding $R(2)$. 
Hence, from now on when we {\em add} a fact, we implicitly add the facts that are required to satisfy full TGDs (i.e., we keep the instance closed under $T_F^*$). 

In the general scenario, because of the complex interaction among facts, the decision whether to add or to discard a fact becomes quite intricate. We illustrate this in the following example.  
\begin{example}
We take the schema $\mathcal{S}=\{R(A,B,C), P(A,B)\}$ with a set of constraints $F=\{R(x,y,z)\rightarrow P(x,y), P\fd{}A\rightarrow B\}$.

First, we consider the instance $I_1=\{P(1,1),R(1,2,1)\}$. We start with an empty instance and begin with the fact $P(1,1)$. We add it as doing so does not violate the constraints. Adding the next fact $R(1,2,1)$ would require adding also the fact $P(1,2)$. This would, however, create a conflict. Hence, we must discard $R(1,2,1)$. The obtained instance $\{P(1,1)\}$ is a repair of $I_1$. 

\rrtext{
\rrmark{Added more details to the example.}
Now, let's consider the instance $I_2=\{R(1,2,1),R(1,2,2)\}$ and begin with the fact $R(1,2,1)$. 
Adding the fact $R(1,2,1)$ would require adding also the fact $P(1,2)$ which is not present in the instance constructed so far. Therefore, we can consider both adding and discarding $R(1,2,1)$. We decide to discard it, but we note that the set of facts $\{P(1,2)\}$ cannot become included later on in the constructed instance. We store it on the list of {\em banned} sets. Intuitively, a banned set contains tuples whose mutual absence in the constructed instance justifies discarding some other tuple. 
Consequently, we must prevent adding any tuples which may cause a banned set to be included in the constructed instance. For instance, adding the next fact $R(1,2,2)$ would require adding also $P(1,2)$, and cause inclusion of the banned set $\{P(1,2)\}$. Hence, we discard $R(1,2,2)$, and finally, obtain the empty repair $\varnothing$. We observe that if the fact $R(1,2,2)$ were added {\em(}together with $P(1,2)${\em)}, the obtained instance $\{R(1,2,2),P(1,2)\}$, although consistent, would not be $\leq_{I_2}$-minimal (the repair $\{R(1,2,1),R(1,2,2),P(1,2)\}$ is relatively closer to $I_2$).
}

Finally, we consider the instance $I_3=\{R(1,1,1), P(1,1), P(1,2)\}$ and begin with the fact $R(1,1,1)$. Discarding it would require memorizing the banned set $\{P(1,1)\}$. However, this set is not appropriate for our purposes because the fact $P(1,1)$ is present in the original instance and later we might be forced to add it to the constructed instance. Hence, we add the fact $R(1,1,1)$ (together with $P(1,1)$). Next, we add the fact $P(1,1)$ but ignore the fact $P(1,2)$. The constructed instance $\{R(1,1,1),P(1,1)\}$ is a repair.
\end{example}

Now, we present a sound and complete Algorithm~\ref{alg:construct-repair} constructing a repair of a (possibly inconsistent) instance $I$ w.r.t. a set $F$ of full TGDs and denial constraints. 
\begin{algorithm}[htb]
\caption{\label{alg:construct-repair} Constructing a repair of $I$ w.r.t. a set of full TGDs $F$}
\begin{tabbing}
mm\=xxx\=xxx\=xxx\=xxx\=xxx\=\kill
{\footnotesize\tt ~1:} \>$I^o\gets I$\\
{\footnotesize\tt ~2:} \>$J\gets\varnothing$\\
{\footnotesize\tt ~3:} \>$Banned\gets\varnothing$\\
{\footnotesize\tt ~4:} \>{\bf while} $I^o\neq\varnothing$ {\bf do}\\
{\footnotesize\tt ~5:} \>\>{\bf choose} $R(t)\in I^o$ {\bf and} $b\in\{\mathbf{true},\mathbf{false}\}$\\
{\footnotesize\tt ~6:} \>\>$I^o\gets I^o\setminus\{R(t)\}$\\
{\footnotesize\tt ~7:} \>\>$J'\gets T^*_F(J\cup\{R(t)\})$\\
{\footnotesize\tt ~8:} \>\>{\bf if} $
\begin{cases}
J'\not\models F \,\text{\bf{}or} 
&\text{(*)}\\
b\land J'\setminus(I\cup J)\neq \varnothing \,\text{\bf{}or} &\text{(**)}\\
\exists B\in Banned\sepdot B\subseteq J' 
&\text{(***)}
\end{cases}
$ {\bf then}\\
{\footnotesize\tt ~9:} \>\>\> {\bf if} $J'\models F$ {\bf then} $Banned\gets Banned\cup\{J'\setminus(I\cup J)\}$\\
{\footnotesize\tt 10:} \>\> {\bf else} $J\gets J'$\\
{\footnotesize\tt 11:} \> {\bf return} $J$
\end{tabbing}
\end{algorithm}
It starts with an empty instance $J$ and iterates over the facts of the original instance $I$. $Banned$ is a collection of {\em banned} sets of facts, i.e. sets that are not to be included in the constructed instance $J$. We note that {\em some} elements of those sets can, however, be included in $J$.

For every fact $R(t)$ the algorithm makes a choice $b$ whether or not it should try discarding the fact. \rrmark{Clarified on the choice of $b$.}\rrtext{Here, this choice is nondeterministic but, in practice, it could be based on the user preference.} The fact $R(t)$ is discarded if one of the following conditions is satisfied:
\begin{itemize}
\item[(*)] Adding $R(t)$ violates the constraints. 
\item[(**)] Adding $R(t)$ does not violate constraints, but $b$ is set to {\bf true} and adding $R(t)$ implies adding facts that are not present in $J$ and $I$. 
\item[(***)] Adding $R(t)$ leads to inclusion of some previously created {\em banned} set. 
\end{itemize}
\rrmark{Added line references.}
If the fact is discarded even though adding it does not violate the constraints, a banned set is added to $Banned$ (\algline{9}). Finally, if none of the conditions above is satisfied, the fact is added to $J$ (\algline{10}).

Before proving that Algorithm~\ref{alg:construct-repair} is sound and complete, we make several observations. First, we note that $J$ is always closed under $T^*_F$. Moreover, $J$ is always consistent because the condition (*) ensures that facts are added to $J$ only if doing so does not violate the constraints. Finally, the main loop of Algorithm~\ref{alg:construct-repair} satisfies the following invariant:
\[
\mathtt{Inv}\equiv\forall B\in Banned\sepdot B\not\subseteq J.
\]
Indeed, the invariant is trivially satisfied before the execution enters the main loop. Also, the condition (***) ensures that facts are added to the constructed instance $J$ only if doing so does not violate $\mathtt{Inv}$. Hence, we need only to check that creating a new banned set does not violate the invariant. We observe that a new banned set is created only if (**) or (***) are satisfied. (**) implies directly that the new banned set satisfies $\mathtt{Inv}$. (***) implies this implicitly. In this case there exists a banned set $B$ such that $B\subseteq J'=T^*_F(J\cup\{R(t)\})$. We note that all banned sets contain no fact from $I$. $B\not\subseteq J$ by $\mathtt{Inv}$. Thus, the newly created banned set $B'=J'\setminus (I\cup J)$ is not included in $J$. 

\begin{theorem}
Algorithm~\ref{alg:construct-repair} is a sound and complete repairing algorithm for any instance and any set of denial constraints and full tuple-generating dependencies. Algorithm~\ref{alg:construct-repair} works in time polynomial in the size of the input instance $I$.
\end{theorem}
\begin{proof}
{\em Soundness}. We show that for any execution of Algorithm~\ref{alg:construct-repair} the returned instance, denoted here by $I'$, satisfies the conditions (i), (ii), and (iii) of Lemma~\ref{lemma:repair-test}. 

The conditions (i) and (ii) are satisfied trivially because, as observed before, at every time the constructed instance $J$ is consistent and closed under $T^*_F$. 

To show (iii) we take any $R(t)\in I\setminus I'$ such that $J'=T^*_F(I'\cup\{R(t)\})\models F$. Consider the iteration of the main loop during which $R(t)$ was chosen and note that $J\subseteq I'$. We observe that the condition (*) is not satisfied, but since $R(t)$ is not added to $I'$, (**) or (***) is. Consequently, a banned set $B=J'\setminus(I\cup J)$ is added to $Banned$. It is easy to see that 
$B\subseteq J'\setminus I$. On the other hand, $\mathtt{Inv}$ implies that $B\not\subseteq I'\setminus I$, which proves that $J'\setminus I\not\subseteq I'\setminus I$.

{\em Completeness}. Take any repair $I'\in Repairs(I,F)$ and consider an execution of Algorithm~\ref{alg:construct-repair} during which: 1) in the first phase, it selects all facts from $I\cap I'$ and sets $b$ to {\bf false}; 2) in the second phase, it chooses all facts from $I\setminus I'$ and sets $b$ to {\bf true}. We show that this execution returns $I'$.

First, we consider the facts chosen in the first phase and we show that none of the conditions (*), (**), and (***) is satisfied. Indeed, the condition (*) is not satisfied because $I'$ is consistent (and so is any of its subsets closed under  $T^*_F$). Trivially, the condition (**) is also not satisfied because $b$ is chosen to be $false$.  The condition (***) is not satisfied because during the first phase $Banned$ remains empty. Note that the instance $J$ obtained after the first phase consists exactly of the facts of $I'\cap I$ closed under $T^*_F$. By (ii) of Lemma~\ref{lemma:repair-test},  $J=I'$. 

Now, we show with a simple inductive argument that none of the facts from $I\setminus I'$ are added in the second phase. By (iii) of Lemma~\ref{lemma:repair-test}, for a fact $R(t)\in I\setminus I'$, if (*) is not satisfied, then $T^*_F(J\cup\{R(t)\})\setminus I\not\subseteq J\setminus I$. This implies that $T^*_F(J\cup\{R(t)\})\setminus (I\cup J)$ is nonempty, i.e. (**) is satisfied ($b$ is {\bf true}).

We finish the proof by showing that Algorithm~\ref{alg:construct-repair} works in time polynomial in the size of $I$. First, we observe the algorithm iterates over the facts of $I$. For every fact it creates at most one banned set, and so the cardinality of $Banned$ is bounded by the size of $I$. Each banned set is a subset of $Hull(I,F)$, and thus of polynomial size as well. Consequently, for every fact all of the conditions (*), (**), and (***) can be checked in polynomial time. 
\qed
\end{proof}

\section{Consistent query answering for full TGDs}
\label{sec:ftgd-cqa}
In this section we investigate consistent query answering in the presence of full tuple-generating dependencies and denial constraints. We begin by presenting a polynomial algorithm for computing consistent answers to quantifier-free queries in the presence of acyclic full TGDs and denial constraints. Next, we extend this approach to handle join dependencies as well. Finally, we show that for arbitrary full TGDs consistent query answering is coNP-complete.

\subsection{Warm-up: acyclic full TGDs and denial constraints}
\label{sec:aftgd}
In this section we extend the algorithm computing consistent query answers to closed quantifier-free queries in the presence of denial constraints~\cite{ChMa04}. The main idea of the algorithm is to check if there exists a repair that does not satisfy the query, i.e. satisfies the negated query. The negated query specifies the facts that need to be present and the facts that need to be absent in the repair. To find if the repair in question can be constructed we devise {\em supports} and {\em blocks} of the facts that need to be respectively present and absent in the repair. Intuitively, a {\em support} of a fact is a set of facts from the original instance that, if contained in the repair, guarantees the presence of this fact. Conversely, a {\em block} of a fact specifies facts that lead to a conflict with this fact and so their presence guarantees the absence of the fact in the repair. Additionally, a block can specify facts that must not be included during the repairing process. Finally, we show how to check that a combination of supports and blocks can be realized in the same repair. 

We fix an instance $I$ and a set $F$ of denial constraints and acyclic full tuple-generating dependencies. For brevity, we use inference-like rules to define the supports and blocks of a fact. A rule of the form $\frac{A}{B}$ reads: ``{\em $B$ provided $A$}''. Also, in $A$ we often use ground rules which implicitly belong to $Rules(I,F)$. 
\begin{definition}[Support]\label{def:support}
A {\em support} of a fact $R(t)\in Hull(I,F)$ is a subset of $I$ defined with the following rules:
\begin{gather*}
\infer[\mathtt{S}_0]{
  R(t)\in I
}{
  \{R(t)\}\in Supp(R(t))
}\\[5pt]
\infer[\mathtt{S}_1]{
  \begin{aligned}
    &&
    &R_1(t_1)\land\ldots\land R_n(t_n)\rightarrow R(t)\\
    R(t)&\not\in I&
    &S_i\in Supp(R_i(t_i))\quad\forall i\in\{1,\ldots,n\}
  \end{aligned}
} {
  \bigcup_i S_i \in Supp(R(t))
}
\end{gather*}
where $Supp(R(t))$ is the set of all supports of $R(t)$.
\end{definition} 
\rrmark{Added an intuitive description of supports.}
\rrtext{Essentially, a support of a tuple from the original instance is a singleton consisting of that tuple (rule $\mathtt{S}_0$). If a tuple does not belong to the original instance but there is a full TGD rule having the tuple in its rhs, then a support of that tuple is a union of the supports of the tuples in the lhs (rule $\mathtt{S}_1$).
}
\begin{example}\label{ex:support}
We take the schema $\mathcal{S}=\{R(A,B,C), P(A,B), Q(A)\}$ with the set of constraints $F=\{R(x,y,z)\rightarrow P(x,y), P(x,y)\rightarrow Q(x), P\fd{}A\rightarrow B\}$, and consider the instance $I=\{R(1,1,1),R(1,2,1), P(1,2), Q(2)\}$. The hull is $Hull(I,F)=I\cup\{P(1,1),Q(1)\}$. The set of ground rules is
\begin{align*}
Rules(I,F)=\{
&R(1,1,1)\rightarrow P(1,1), R(1,2,1)\rightarrow P(1,2), \\
&P(1,1)\rightarrow Q(1), P(1,2)\rightarrow Q(1),\\
&P(1,1)\land P(1,2)\rightarrow \mathbf{false}\}.
\end{align*}
$Repairs(I,F)$ consists of the following instances:
\begin{gather*}
I_1'=\{Q(2)\},\quad
I_2'=\{R(1,1,1),P(1,1),Q(1), Q(2)\},\\
I_3'=\{R(1,2,1),P(1,2),Q(1), Q(2)\}.
\end{gather*}
The facts from $I$ have simple supports obtained with the rule  $\mathtt{S}_0$:
\begin{align*}
Supp(R(1,1,1))&=\{\{R(1,1,1)\}\},&
Supp(P(1,2))&=\{\{P(1,2)\}\},\\
Supp(R(1,2,1))&=\{\{R(1,2,1)\}\},&
Supp(Q(2))&=\{\{Q(2)\}\}.
\end{align*}
The fact $P(1,1)$ has only one support obtained with the rule $\mathtt{S}_1$:
\begin{gather*}
Supp(P(1,1))=\{\{R(1,1,1)\}\}.
\end{gather*}
Finally, $Q(1)$ has two supports:
\begin{gather*}
Supp(Q(1))=\{\{R(1,1,1)\},\{P(1,2)\}\}.
\end{gather*}
\end{example}
The supports of a fact define conditions that ensure that it is present in the repair.  
\begin{proposition}\label{prop:support}
For every $I'\in Repairs(I,F)$ and every $R(t)\in Hull(I,F)$ 
\[
R(t)\in I' \iff \exists S\in Supp(R(t))\sepdot S\subseteq I'.
\]
\end{proposition}
The proof is by a simple induction over the position of the relation name $R$ in a topological sort of the dependency graph $\mathcal{D}(\mathcal{S},F)$. The proof of a more general claim  (Proposition~\ref{prop:jd-support}) can be found in Appendix~\ref{app:omitted-proofs}.

Blocking a fact is more complex. First, the facts that are not present in the original instance $I$ need to be added in the process of creating a repair. In this case we can explicitly forbid adding this fact (rule $\mathtt{B}_0$). For instance, the fact $Q(1)$ can be blocked this way. Facts that belong to $I$ can be blocked using conflicts they are involved in. If a fact is involved in a denial conflicts then it is blocked by the presence of other facts that together lead to a conflict (rule $\mathtt{B}_1$). Finally, blocks can be propagated using full TGDs (rule $\mathtt{B}_2$).

Consequently, a block of a fact consists of two sets: one indicating the facts from $I$ that need to be present in the repair and the other one indicating a fact that must not be added to the repair. 
\begin{definition}[Block]\label{def:block}
A {\em block} of a fact $R(t)\in Hull(I,F)$ is a pair that consists of a subset of $I$ and a set of at most one fact from $Hull(I,F)\setminus I$, defined with the following rules: 
\begin{gather*}
\infer[\mathtt{B}_0]{
  R(t)\not\in I
}{
  (\varnothing,\{R(t)\})\in Block(R(t))
}\\[5pt]
\infer[\mathtt{B}_1]{
  \begin{aligned}
    &&
    & R(t)\land R_1(t_1)\land\ldots\land R_n(t_n)\rightarrow \mathbf{false}\\
    &R(t)\in I&
    & S_i\in Supp(R_i(t_i)) \quad\forall i \in \{1,\ldots,n\}
  \end{aligned}
} {
  (\bigcup_i S_i,\varnothing)\in Block(R(t))
}\\[5pt]
\infer[\mathtt{B}_2]{
  \begin{aligned}
    &&
    & R(t)\land R_1(t_1)\land\ldots\land R_n(t_n)\rightarrow P(s)\\
    &&
    & S_i\in Supp(R_i(t_i)) \quad\forall i \in \{1,\ldots,n\}\\
    & R(t)\in I&
    & (B,N)\in Block(P(s))
  \end{aligned}
} {
  (\bigcup_i S_i\cup B,N)\in Block(R(t))
}
\end{gather*}
where $Block(R(t))$ being the set of all blocks of $R(t)$.
\end{definition}
Blocks specify the conditions that ensure a fact to be absent in a repair. 
\begin{proposition}\label{prop:block}
For every $I'\in Repairs(I,F)$ and every $R(t)\in Hull(I,F)$
\[
R(t)\not\in I' \iff \exists (B,N)\in Block(R(t))\sepdot
B\subseteq I' \land N\cap I'=\varnothing. 
\]
\end{proposition}
The proof is by induction over the position of $R$ in a reverse topological sorting of $\mathcal{D}(\mathcal{S},F)$. The proof of a more general claim  (Proposition~\ref{prop:jd-block}) can be found in Appendix~\ref{app:omitted-proofs}. We remark, however, that acyclicity of $F$ is essential here. 
\begin{example}[cont. Example~\ref{ex:support}]\label{ex:block}
The facts $Q(1)$ and $P(1,1)$ have simple blocks obtained with the rule $\mathtt{B}_0$:
\begin{align*}
Block(Q(1))&=\{(\varnothing,\{Q(1)\})\}&
Block(P(1,1))&=\{(\varnothing,\{P(1,1)\})\}.
\end{align*}
The fact $R(1,1,1)$ has one block obtained with the rule $\mathtt{B}_2$:
\begin{gather*}
Block(R(1,1,1))=\{(\varnothing,\{P(1,1)\})\}.
\end{gather*}
The fact $P(1,2)$ has two blocks obtained with the rules $\mathtt{B}_1$ and $\mathtt{B}_2$:
\begin{gather*}
Block(P(1,2))=\{(\{R(1,1,1)\},\varnothing),(\varnothing,\{Q(1)\})\}.
\end{gather*}
The fact $R(1,2,1)$ has two blocks obtained with the rule $\mathtt{B}_2$:
\begin{gather*}
Block(R(1,2,1))=\{(\{R(1,1,1)\},\varnothing),(\varnothing,\{Q(1)\})\}.
\end{gather*}
Finally, the fact $Q(2)$ has no blocks (it does not participate in any conflict):
\begin{gather*}
Block(Q(2))=\varnothing.
\end{gather*}
\end{example}
The following proposition ensures tractability of our approach. 
\begin{proposition}\label{prop:support-and-block-size}
For any fact the number of all of its supports and the number of all of its blocks can be computed in time polynomial in the size of the instance $I$. 
\end{proposition}
This claim is proved with a simple combinatorial argument. Also here, the acyclicity of $F$ is essential. The proof of a more general claim (Proposition~\ref{prop:jd-support-and-block-size}) can be found in Appendix~\ref{app:omitted-proofs}. 

Finally, we show how to check if there exists a repair that realizes a given combination of supports and blocks. 
\begin{lemma}\label{lemma:base-repair-test}
For any (possibly cyclic) set of full TGDs and denial constraints $F$, an instance $I$, and two sets of facts $P\subseteq I$ and $N\subseteq Hull(I,F)\setminus I$, a repair containing all facts from $P$ and no facts from $N$ exists if and only if $T^*_F(P)$ is consistent and disjoint with $N$. 
\end{lemma}
\begin{proof}
The {\em only if\/} part of the proof is trivial. 
For the {\em if\/} part take any repair $I'$ such that $I'\leq_I T^*_F(P)$. Such an instance exists because $T^*_F(P)$ is consistent. We show that $I'$ is the desired repair. First, $I'\leq_I T^*_F(P)$ and $P\subseteq I$ imply that $P\subseteq I'$, and consequently $T^*_F(P)\subseteq I'$ (as $I'$ is consistent). $I'\leq_I T^*_F(P)$ implies also $N\cap I'=\varnothing$ because $N$ contains no fact from $I$. \qed
\end{proof}

We use the previous results to construct Algorithm~\ref{alg:cqa} computing the consistent answer to a quantifier-free query to $Q$ in the instance $I$ w.r.t. a set of denial constraints and acyclic full TGDs $F$. 
\begin{algorithm}[htb]
\caption{\label{alg:cqa} Computing the consistent answer to $Q$ in $I$ w.r.t. $F$.}
\begin{tabbing}
mm\=xxx\=xxx\=xxx\=xxx\=xxx\=\kill
{\bf function} {\sc CQA($Q$, $I$, $F$)}\\
{\bf precompute:} $Hull(I,F)$, $Supp$, and $Block$ for $I$ and $F$. \\
{\footnotesize\tt ~1:}\>{\bf let} $Q\equiv Q_1\land\ldots\land Q_w$ \qquad{\it /*~Query in CNF~*/}\\
{\footnotesize\tt ~2:}\>{\bf for} $i \gets 1,\ldots, w$ {\bf do}\\
{\footnotesize\tt ~3:}\>\>{\bf let} $\neg Q_i\equiv R_1(t_1)\land\ldots\land R_n(t_n) \land \neg P_{1}(s_{1})\land\ldots\land \neg P_{m}(s_{m})$\\
{\footnotesize\tt ~4:}\>\>{\bf if} {\sc existsRepair($\{R_1(t_1),\ldots,R_n(t_n)\}$, $\{P_{1}(s_{1}),\ldots,P_{m}(s_{m})\}$)} {\bf then}\\
{\footnotesize\tt ~5:}\>\>\>{\bf return false}\\
{\footnotesize\tt ~6:}\>{\bf return true}\\
{\bf end function}\\
\\
{\bf function} {\sc existsRepair($T_P$, $T_N$)}\\
{\footnotesize\tt ~7:} \> {\bf if} $T_P\not\subseteq Hull(I,F)$ {\bf then}\\
{\footnotesize\tt ~8:} \>\> {\bf return} {\bf false}\\
{\footnotesize\tt ~9:} \> {\bf if} $T_N\not\subseteq Hull(I,F)$ {\bf then}\\
{\footnotesize\tt 10:} \>\> $T_N\gets T_N\cap Hull(I,F)$\\
{\footnotesize\tt 11:} \> {\bf let} $\{R_1(t_1),\ldots,R_n(t_n)\}\equiv T_P$\\
{\footnotesize\tt 12:} \> {\bf let} $\{P_1(s_1),\ldots,P_m(t_m)\}\equiv T_N$\\
{\footnotesize\tt 13:} \> {\bf for} $S_1\in Supp(R_1(t_1)),\ldots,S_n\in Supp(R_n(s_n)$ {\bf do}\\
{\footnotesize\tt 14:} \>\> {\bf for} $(B_1,N_1)\in Block(P_1(s_1)),\ldots,(B_m,N_m)\in Block(P_m(s_m))$ {\bf do}\\
{\footnotesize\tt 15:} \>\>\> $P\gets S_1\cup\ldots\cup S_n \cup B_1\cup\ldots\cup B_m$\\
{\footnotesize\tt 16:} \>\>\> $N\gets N_1\cup\ldots\cup N_m$\\
{\footnotesize\tt 17:} \>\>\> {\bf if} $T^*_F(P)\models F$ {\bf and} $T^*_F(P)\cap N=\varnothing$ {\bf then}\\
{\footnotesize\tt 18:} \>\>\>\> {\bf return true}\\
{\footnotesize\tt 19:} \> {\bf return false}\\
{\bf end function}
\end{tabbing}
\end{algorithm}

We assume that the query $Q$ is in CNF and we note that {\bf true} is {\em not} the consistent answer to $Q$ if and only if there exists a conjunct of $Q$ that is {\em not} satisfied by some repair.
\rrmark{Added line references.}
Consequently, for each conjunct $Q_i$ of $Q$ we check if there exists a repair that satisfies $\neg Q_i$ (\algline{2}). A negated conjunct is a conjunction of positive and negative atomic formulas 
\[
\neg Q_i \equiv 
R_1(t_1)\land\ldots\land R_n(t_n)\land 
\neg P_1(s_1)\land\ldots\land P_m(t_m)
\]
and therefore a repair satisfying $\neg Q_i$ is a repair that contains all  $R_i(t_i)$'s and no $P_j(s_j)$'s. The existence of such a repair is checked with the function {\sc ExistsRepair}. 

Because all repairs are constructed from facts in $Hull(I,F)$, we can assume that all $R_i(t_i)$'s and $P_j(s_j)$'s belong to $Hull(I,F)$. Indeed, if some $R_i(t_i)$ does not belong to $Hull(I,F)$, then a repair containing $R_i(t_i)$ does not exist (\algline{7}). Similarly, if some $P_j(s_j)$ does not belong to $Hull(I,F)$, then no repair contains $P_j(s_j)$ (\algline{9}). Using Propositions~\ref{prop:support} and~\ref{prop:block} we show that there exists a repair $I'$ containing all $R_i(t_i)$'s and no $P_j(s_j)$'s if and only if for every $R_i(t_i)$ there exists a support $S_i$ and for every $P_j(s_j)$ there exists a block $(B_j,N_j)$ such that $I'$ contains all $S_i$'s and $B_j$'s and is disjoint with every $N_j$'s. 

Hence, it suffices to exhaustively enumerate over all combinations of supports of $R_i(t_i)$'s (\algline{13}) and blocks of $P_j(s_j)$'s (\algline{14}) and use Lemma~\ref{lemma:base-repair-test} to check if a combination can be realized by a repair (\algline{17}).

Finally, to show that Algorithm~\ref{alg:cqa} works in time polynomial in the size of $I$, we note that the size of the query is considered to be a fixed constant and by Proposition~\ref{prop:support-and-block-size} the number of supports and blocks of every fact is polynomial in the size of $I$. We remark that acyclicity is used only in Propositions~\ref{prop:block} and~\ref{prop:support-and-block-size}.

\begin{theorem}\label{thm:acyclic-ftgd-tractable}
Consistent query answering is in PTIME for any quantifier-free query and any acyclic set of denial constraints and full tuple-generating dependencies.
\end{theorem}
\begin{example}[cont. Example~\ref{ex:block}]
We execute Algorithm~\ref{alg:cqa} with the following query:
\[
Q=(Q(1)\lor \neg R(1,1,1))\land (Q(2)\lor \neg P(1,2)) 
\land (R(1,2,1)\lor\neg P(1,2)).
\]
The negation of the first conjunct is $\neg Q_1=R(1,1,1)\land\neg Q(1)$. The fact $R(1,1,1)$ has only one support $\{R(1,1,1)\}$ and the fact $Q(1)$ only one block $(\varnothing,\{Q(1)\})$. Although $T_F^*(\{R(1,1,1)\})=\{R(1,1,1),P(1,1),Q(1)\}$ is consistent, it contains $Q(1)$. Hence, there does not exists a repair that satisfies $\neg Q_1$.

The negation of the second conjunct is $\neg Q_2=P(1,2)\land\neg Q(2)$. Because the fact $Q(2)$ has no block, there is no repair that does not contain $Q(2)$, and consequently there does not exists a repair satisfying $\neg Q_2$.

The negation of the third conjunct is $\neg Q_3=P(1,2)\land\neg R(1,2,1)$. The fact $P(1,2)$ has only one support $\{R(1,1,1)\}$ and  the fact $R(1,2,1)$ has two blocks: $(\varnothing,\{Q(1)\})$ and $(B_2,N_2)=(\{R(1,1,1)\},\varnothing)$. Similarly to $\neg Q_1$, combining the support with the first block does not guarantee the existence of a repair satisfying $\neg Q_3$. However, if we use the support with the second block, then $P=\{R(1,1,1)\}$ and $N=\varnothing$ satisfy Lemma~\ref{lemma:base-repair-test} which implies that there exists a repair satisfying $\neg Q_3$. Indeed, this repair is $I_2'$.

Consequently, the query $Q$ does not hold in the repair $I_2'$ and {\bf true} is not the consistent answer to $Q$. 
\end{example}
\subsection{Adding join dependencies}
In this section we extend Algorithm~\ref{alg:cqa} to include also JDs. For this we generalize the definitions of supports and blocks and we show that Propositions~\ref{prop:support},~\ref{prop:block}, and~\ref{prop:support-and-block-size} continue to hold. Here, we present only the constructions and main claims. Complete proofs are presented in Appendix~\ref{app:omitted-proofs}.

The following folklore result shows that we need to consider the case where there is only one JD per relation.
\begin{proposition}\label{prop:one-jd-per-relation}
For any (possibly empty) set of JDs $\{jd_1,\ldots,jd_n\}$ on the same relation name there exists a JD $jd^*$ such that an instance satisfies $\{jd_1,\ldots,jd_n\}$ if and only if it satisfies $jd^*$.
\end{proposition}
We remark that in particular every relation $R$ in every instance satisfies the {\em trivial} JD $R\jd[attrs(R)]$. Hence, we fix an instance $I$ and a set of constraints $F$ consisting of a set of acyclic full TGDs, denial constraints, and exactly one JD per relation. Also, to distinguish JD rules we write them $R(t_1)\land\ldots\land R(t_n)\jdrightarrow R(t)$. We remark, however,  that $\rightarrow$ continues to be used for all rules, including those obtained from conflicts w.r.t. JDs. Finally, we observe that any JD $R\jd[X_1,\ldots,X_k]$ yields rules $R(t)\jdrightarrow R(t)$ for all $R(t)\in Hull(I,F)$. Because we assume that $F$ contains one JD on every relation in $\mathcal{S}$,  $\jdrightarrow$ and $\rightarrow$ (which subsumes $\jdrightarrow$) are {\em reflexive} on facts from $Hull(I,F)$.

Previously, the acyclicity of the set of integrity constraints implicitly provided a bound on the depth of the derivation of supports and blocks. This bound is essential when showing that supports and blocks can be constructed in time polynomial in the size of the database. Although JDs translate to cyclic full TGDs, it is sufficient to consider derivations of supports and blocks of bounded depth. 
\rrmark{Moved {\em unfolding} rules and the proof of Lemma 3 to the appendix.}
\begin{lemma}\label{lemma:jd-weaving}
If $Rules(I,F)$ contains the following two ground rules
\[
r'=R(t'_1)\land\ldots\land R(t'_k)\jdrightarrow R(t_i)
\quad\text{and}\quad
r''=R(t_1)\land\ldots\land R(t_k)\jdrightarrow R(t)
\]
for some $i\in\{1,\ldots,k\}$, then there exists $j\in\{1,\ldots,k\}$ such that $Rules(I,F)$ contains also 
\[
r^*=R(t_1)\land\ldots\land R(t_{i-1})\land R(t'_j)\land R(t_{i+1})\land\ldots\land R(t_k)\jdrightarrow R(t).
\]
\end{lemma}

Because the set of constraints is cyclic, we construct the supports in an iterative manner allowing us to bound the derivation depth. \rrtext{The new rules for supports are obtained by appropriately  incorporating JD-rules into $\mathtt{S}_0$ and $\mathtt{S}_1$ (Definition~\ref{def:support}).}
\begin{definition}[Support]
Let $h$ be the acyclic height of the dependency graph $\mathcal{D}(\mathcal{S},F)$. For $\ell\in \{-1,0,\ldots,h\}$ an {\em $\ell$-support} of a fact $R(t)\in Hull(I,F)$ is a subset of $I$ defined with the following rules:
\begin{gather*}
\infer[\mathtt{S}_0^{-1}]{
R(t)\in I
}{
\{R(t)\}\in Supp^{-1}(R(t))
}
\\[10pt]
\infer[\mathtt{S}_1^\ell]{
&&&R(t_1)\land\ldots\land R(t_k)\jdrightarrow R(t)\\
&&&R_{i,1}(t_{i,1})\land\ldots\land R_{i,n_i}(t_{i,n_i})\rightarrow R(t_i)
\quad\forall i\in\{1,\ldots,k\}\\
&R(t)\not\in I& &S_{i,j}\in Supp^{\ell-1}(R_{i,j}(t_{i,j}))\quad
\forall i\in\{1,\ldots,k\},\;\forall j\in\{1,\ldots,n_i\}
}{
\bigcup_{i,j} S_{i,j}\in Supp^\ell(R(t))
}
\end{gather*}
where $Supp^\ell(R(t))$ denotes the set of all $\ell$-supports of $R(t)$. A support of $R(t)$ is any element of the set $Supp(R(t))=Supp^{h}(R(t))$.
\end{definition}
We note that because $\rightarrow$ and $\jdrightarrow$ are reflexive on facts from $Hull(I,F)$, the rule $\mathtt{S}_1^\ell$ properly propagates supports, i.e any $(\ell-1)$-support of $R(t)$ is also an $\ell$-support of $R(t)$. We also note that if all the JDs in $F$ are trivial, then the set of supports of a fact coincides with the set of supports from Definition~\ref{def:support}

\begin{proposition}~\label{prop:jd-support}
For every $I'\in Repairs(I,F)$ and every $R(t)\in Hull(I,F)$
\[
R(t)\in I'\iff \exists S\in Supp(R(t)). S\subseteq I'\sepdot
\]
\end{proposition}
The {\em if\/} part is proved with a simple induction over the depth of the derivation of a support. The proof of the {\em only if\/} is based on the following simple idea. A fact $R(t)\in I'\setminus I$ is present in the repair $I'$ to satisfy some ground (full TGD) rule. We identify this ground rule by considering an inconsistent instance $I'\setminus\{R(t)\}$. We use this rule to show that $R(t)$ has a support that validates the claim.

Again, because the set of constraints is cyclic, we construct the block iteratively to bound their derivation depth. \rrtext{The new rules for block are obtained by appropriately incorporating JD-rules into $\mathtt{B}_0$, $\mathtt{B}_1$, and $\mathtt{B}_2$ (Definition~\ref{def:block}).}
\begin{definition}[Block]
Let $h$ be the acyclic height of the dependency graph $\mathcal{D}(\mathcal{S},F)$. For $\ell\in\{-1,0,\ldots,h\}$ an {\em $\ell$-block} of a fact $R(t)\in Hull(I,F)$ is a pair that consists of a subset of $I$ and a set of at most one fact from $Hull(I,F)\setminus I$, defined with the following rules: 
\begin{gather*}
\infer[\mathtt{B}_0^{-1}]{
R(t)\not\in I
}{
(\varnothing,\{R(t)\})\in Block^{-1}(R(t))
}
\\[10pt]
\infer[\mathtt{B}_1^{-1}]{
&&&R(t_1)\land\ldots\land R(t_n)\land 
R_1(s_1)\land\ldots\land R_m(s_m)\rightarrow \mathbf{false}\\
&&&R(t)\land R(t_{i,1})\land\ldots\land R(t_{i,k_i})\jdrightarrow R(t_i)
\quad\forall i\in\{1,\ldots,n\}\\
&&&S_{i,j}\in Supp(R(t_{i,j}))
\quad\forall i\in\{1,\ldots,n\},\;\forall j\in\{1,\ldots,k_i\}\\
&R(t)\in I&& S_p\in Supp(R_p(s_p))\quad\forall p\in\{1,\ldots,m\}
}{
(\bigcup_{i,j} S_{i,j}\cup\bigcup_p S_p,\varnothing)\in Block^{-1}(R(t))
}\\[10pt]
\infer[\mathtt{B}_2^\ell]{
&&&R(t_1)\land\ldots\land R(t_n)\land 
R_1(s_1)\land\ldots\land R_m(s_m)\rightarrow P(s)\\
&&&R(t)\land R(t_{i,1})\land\ldots\land R(t_{i,k_i})\jdrightarrow R(t_i)
\quad\forall i\in\{1,\ldots,n\}\\
&&&S_{i,j}\in Supp(R(t_{i,j}))
\quad\forall i\in\{1,\ldots,n\},\;\forall j\in\{1,\ldots,k_i\}\\
&&& S_p\in Supp(R_p(t_p))\quad\forall p\in\{1,\ldots,m\}\\
&R(t)\in I& &(B,N)\in Block^{\ell-1}(P(s))
}{
(\bigcup_{i,j} S_{i,j}\cup\bigcup_p S_p\cup B,N)\in Block^\ell(R(t))
}
\end{gather*}
where $Block^\ell(R(t))$ is the set of all $\ell$-blocks of $R(t)$. A {\em block} of $R(t)$ is any element of the set $Block(R(t))=Block^h(R(t))$.
\end{definition}
Also this time, we observe that because $\rightarrow$ and $\jdrightarrow$ are reflexive, the rule $\mathtt{B}_1^\ell$ properly propagates blocks, i.e. any $(\ell-1)$-block of $R(t)$ is also an $\ell$-block of $R(t)$. We also note that this definition of blocks generalizes Definition~\ref{def:block}.
\begin{proposition}\label{prop:jd-block}
For every $I'\in Repairs(I,F)$ and every $R(t)\in Hull(I,F)$ 
\[
R(t)\not\in I'\iff \exists (B,N)\in Block(R(t))\sepdot B\subseteq I' \land N\cap I'=\varnothing.
\]
\end{proposition}
The {\em if\/} part is proved with a simple induction over the depth of derivation of a block. The proof of the {\em only if\/} part is based on the following simple idea. $R(t)\in I\setminus I'$ is absent in the repair $I'$ because its presence would cause a violation of some ground rule. We identify this rule by considering an inconsistent instance $I'\cup\{R(t)\}$. We use this rule to show that $R(t)$ has a block that validates the claim. 
\begin{proposition}\label{prop:jd-support-and-block-size}
For any fact $R(t)$ the sets $Supp(R(t))$ and $Block(R(t))$ can be constructed in time polynomial in the size of $I$. 
\end{proposition}
A simple combinatorial proof is presented in Appendix~\ref{app:omitted-proofs}. 

We recall that the proof of Theorem~\ref{thm:acyclic-ftgd-tractable} relies on the Lemma~\ref{lemma:base-repair-test} and Propositions~\ref{prop:support},~\ref{prop:block}, and~\ref{prop:support-and-block-size}. The proof of Lemma~\ref{lemma:base-repair-test} does not assume the set of constraints to be acyclic and the corresponding  Propositions~\ref{prop:jd-support},~\ref{prop:jd-block}, and~\ref{prop:jd-support-and-block-size} have been proved for  generalized supports and blocks. Consequently,
\begin{corollary}
Consistent query answering is in PTIME for any quantifier-free query and any set of join dependencies, denial constraints, and acyclic full tuple-generating dependencies.
\end{corollary}
\subsection{Negative results}
It appears to be difficult to extend our approach beyond quantifier-free queries because of the following result.
\begin{theorem}[\cite{ChMa04}] \label{thm:cqa-denial-intractable}
There exists an FD and a closed conjunctive query (using existential quantifiers) for which consistent query answering is coNP-complete.
\end{theorem}
Also, the class of constraints is likely to be maximal as having even one cyclic full TGD that is not a JD leads to intractability. 
\begin{theorem}
There exists a positive atomic query and a set of integrity constraints consisting of one FD and one cyclic full tuple-generating dependency for which consistent query answering is coNP-complete. 
\end{theorem}
\begin{proof}
The membership of consistent query answering in coNP follows from the definition of consistent query answers and Theorem~\ref{thm:rcheck-cyclic}.

We show coNP-hardness by reducing the complement of 3COL to consistent query answering. 3COL is a classic NP-complete problem of testing if a graph has a legal 3-coloring~\cite{Papa94}. A {\em 3-coloring} is an assignment of one of 3 colors to each vertex of the graph. It is {\em legal} if no two adjacent vertices have the same color. Take any undirected graph $G=(V,E)$, and let $V=\{v_1,\ldots,v_n\}$ and $E=\{e_1,\ldots,e_m\}$. We assume that $G$ has no isolated vertices (i.e. vertices incident to no edge). 

We use the schema
\[
\mathcal{S}=\{R(A,B,C,D),P(C)\}
\]
with the set of integrity constraints
\[
F=\{R\fd{}A\rightarrow B, 
R(x_1,y_1,z_1,z_2)\land 
R(x_2,y_2,z_1,z_2)\land 
P(z_1) 
\land y_1\neq y_2\rightarrow P(z_2)\}.
\]
We use the following facts:
\begin{itemize}
\item $p_{i,j}^k=R(i,k,j-1,j)$ for each vertex $v_i$ with color $k$ incident to the edge $e_j$ (we create a separate copy for each edge incident to the vertex and for each color);
\item $q_j=P(j)$ indicating that the edges $e_1,\ldots,e_j$ connect properly colored vertices (for $j\in\{0,\ldots,m\}$);
\item 3 special facts: $r=R(n+1,0,m,m+1)$, $r'=R(n+2,1,m,m+1)$ and $r''=P(m+1)$.
\end{itemize}
The constructed instance is:
\[
I_G=\{p_{i,j}^k\sepbar v_i\in V, e_j\in E, v_i\in e_j, 1\leq k \leq 3\}\cup\{q_0,r,r'\}.
\]

Now, we outline the interaction among the facts induced by the integrity constraints. The FD ensures that every vertex has at most one color assigned to it, i.e. for any $v_i\in V$, any two $e_{j_1},e_{j_2}\in E$ adjacent to $v_i$, and any two {\em different} colors $k_1,k_2\in\{1,2,3\}$:
\begin{equation}
p_{i,j_1}^{k_1}\land p_{i,j_2}^{k_2}\rightarrow \mathbf{false}.
\label{eq:fd-rules}
\end{equation}
The full TGD ensures that the facts $q_j$ are properly used to incrementally verify that all edges connect legally colored vertices, i.e. for any edge $e_j=\{v_{i_1},v_{i_2}\}$ and any two {\em different} colors $k_1,k_2\in\{1,2,3\}$:
\[
p_{i_1,j}^{k_1}\land p_{i_2,j}^{k_2}\land q_{j-1}\rightarrow q_j.
\]
The full TGD also requires that if $q_m$ is inserted, which indicates a legal coloring, then $r$ or $r'$ is to be deleted or $r''$ is to be inserted:
\[
r\land r'\land q_m\rightarrow r''.
\]
Consequently, the query used in the reduction checks if $r$ is not removed from any of the repairs. 
\[
Q=r.
\]

The main claim is that:
\[
G\in \text{3COL}\iff\exists I'\in Repairs(I_G,F)\sepdot r\not\in I'.
\]

For the {\em only if\/} part, let $f$ be the legal $3$-coloring of $G$. We construct the following instance 
\[
I'=\{p_{i,j}^{f(i)} \sepbar v_i\in V, e_j\in E, v_i\in e_j\}\cup\{q_0,\ldots,q_m,r'\}.
\]
It can be easily verified that this instance satisfies (i), (ii), and (iii) of Lemma~\ref{lemma:repair-test} and hence $I'$ is a repair and $r\not\in I'$. 

For the {\em if\/} part, we note that since $r\not\in I'$, by (iii) of Lemma~\ref{lemma:repair-test} both $q_m$ and $r'$ are present in $I'$ and $r''\not\in I'$. $q_m\in I'$ implies that $\{q_0,\ldots,q_m\}\subseteq I'$. Therefore, for every $j\in\{1,\ldots,m\}$ if $e_j=\{v_{i_1},v_{i_2}\}$, there exist two different colors $k_1$ and $k_2$ such that $p_{i_1,j}^{k_1}$ and $p_{i_2,j}^{k_2}$ are present in $I'$. Moreover, for any two $p_{i,j_1}^{k_1}$ and $p_{i,j_2}^{k_2}$ we have $k_1=k_2$. Hence, the function 
$
f(i)= k \quad \text{such that there exists $p_{i,j}^k\in I'$}
$
is a properly defined legal $3$-coloring of $G$. 
\qed
\end{proof}
Finally, we note that the role of the FD can be simulated with a full TGD giving almost the same reduction. Indeed, if we replace the FD with the TGD
$
R(x,y,z_1,z_2)\land R(x,y',z_1',z_2')\rightarrow R(0,0,0,0)
$
and by $d$ denote the fact $R(0,0,0,0)$, then the rules~\eqref{eq:fd-rules} are replaced by 
$
p_{i,j_1}^{k_1}\land p_{i,j_2}^{k_2}\rightarrow d,
$
for every $v_i\in V$, any two edges $e_{j_1}, e_{j_2}$ adjacent to $v_i$, and any two {\em different} colors $k_1,k_2\in\{1,2,3\}$. We note that the fact $d$ is not involved in any other rules, and therefore it is only present in a repair which assign two different colors to the same vertex. Now, the query needs to be augmented to check that in no repair $r$ is deleted while $d$ is not inserted, i.e. in all repairs $r$ is absent only if $d$ is present, $Q=(\neg r) \Rightarrow d = r\lor d$.
\begin{corollary}
There exists a quantifier-free ground query and a set of two full cyclic TGDs for which consistent query answering is coNP-complete.
\end{corollary}
The complexity of computing consistent answers to {\em atomic} ground queries in the presence of full TGDs only remains an open question. 
\section{Universal constraints}
\label{sec:univ}
In this section we investigate the complexity of consistent query answering and repair checking in the presence of arbitrary universal constraints. 

\begin{lemma}
\label{lemma:cqa-upper-bound}\label{lemma:repair-check-upper-bound}
For any set of universal constraints $F$ and any closed query $Q$, repair checking is in coNP and consistent query answering is in $\Pi^p_2$. 
\end{lemma}
\begin{proof}
We observe that checking if a set of facts is a maximal independent set is in PTIME. The definition of a nondeterministic Turing machine checking if an instance $I'$ is not a repair follows from Propositions~\ref{prop:saturated} and~\ref{prop:repair-mis}. First, the machine constructs $Compl(I')$ and checks if it is a maximal independent set. If so, it nondeterministically attempts to construct a maximal independent set $N$ such that $N^+<_I I'$. 

The definition of a nondeterministic machine (with an NP oracle) that checks if {\bf true} is not the consistent query answer follows from Definition~\ref{def:cqa}: {\bf true} is not the consistent answer if and only if there exists a repair where the query answer is {\bf false}. Hence, the machine nondeterministically creates an instance $I'$, verifies that $I'$ is a repair by using the NP oracle (checking that $I'$ is not accepted by the machine constructed above), and verifies that the query answer in $I'$ is {\bf false}. 
\qed
\end{proof}

\begin{theorem}\label{thm:universal-cqa-intractable}
There exists a positive atomic query, and a set of two FDs and a universal constraint for which consistent query answering is  $\Pi^p_2$-complete. 
\end{theorem}
\begin{proof}
The membership is proved in Lemma~\ref{lemma:cqa-upper-bound}. We prove $\Pi_2^p$-hardness by reducing the problem of validity of $\forall^*\exists^*$QBF to $\mathcal{D}_{F,Q}$. 

Consider the following $\forall^*\exists^*$QBF formula:
\[
\Psi=\forall{}x_1,\ldots{},x_n\sepdot\exists{}x_{n+1},\ldots{},x_{n+m}\sepdot\Phi,
\]
where $\Phi=C_1 \land \ldots{} \land C_k$ is quantifier-free 3CNF
i.e., $C_j$ is a clause of three literals 
$L_{j,1} \lor L_{j,2} \lor L_{j,3}$. Recall that checking the validity of $\forall^*\exists^*$QBF is a classical $\Pi_2^p$-complete problem \cite{Papa94}.

We assume that no two clauses of $\Phi$ are identical. For ease of reference, we call the variables $x_1,\ldots,x_n$ {\em universal} and the variables $x_{n+1},\ldots,x_{n+m}$ {\em existential}. We also use the following functions on the literals of $\Phi$:
\begin{align*}
&
\begin{array}{rl}
var(x_i)&=i,\\
var(\neg x_i)&=i,
\end{array}
&
&\begin{array}{rl}
sgn(x_i)&=1,\\
sgn(\neg x_i)&=0,
\end{array}
&
q(x_i)=q(\neg x_i)&=
\begin{cases}
1 &\text{if $i\leq n$,}\\
0 &\text{otherwise.}
\end{cases}
\end{align*}

The schema contains two relation names: 
\[
\mathcal{S}=\{ R(A_1,B_1,A_2,B_2),D(A_1,B_1,C_1,D_1,\ldots,A_4,B_4,C_4,D_4)\}.
\]
The set of integrity constraints is:
\begin{align*}
F=\{&R\fd{}A_1\rightarrow B_1, R\fd{}A_2\rightarrow B_2,\\
    &D(\bar{x}_1,\bar{x}_2,\bar{x}_3,\bar{x}_4)\rightarrow
     R(\bar{x}_1)\lor
     R(\bar{x}_2)\lor
     R(\bar{x}_3)\lor
     R(\bar{x}_4)\},
\end{align*}
where each $\bar{x}_i$ is a vector of 4 variables. 
We use the following types of facts in the reduction:
\begin{itemize}
\item facts corresponding to the valuations of universal variables
($i\in\{1,\ldots,n\}$):
\[
p_i=R(i,1,0,0)\quad\text{and}\quad\bar{p}_i=R(i,0,0,0),
\]
and the valuations of existential variables ($i\in\{n+1,\ldots,n+m\}$):
\[
p_i=R(i,1,1,0)\quad\text{and}\quad\bar{p}_i=R(i,0,1,0),
\]
\item facts corresponding to the clauses ($j\in\{1,\ldots,k\}$):
\begin{align*}
q_j=D(&var(l_{j,1}),sgn(l_{j,1}),q(l_{j,1}),0,
      var(l_{j,2}),sgn(l_{j,2}),q(l_{j,2}),0,\\
      &var(l_{j,3}),sgn(l_{j,3}),q(l_{j,3}),0,0,1,1,1),
\end{align*}
\item 2 special facts:
\[
r=R(0,1,1,1)\quad\text{and}\quad\bar{r}=R(0,0,0,0).
\]
\end{itemize}
The constructed instance is
\[
I_\Psi=\{p_1,\bar{p}_1,\ldots,p_{n+m},\bar{p}_{n+m},
q_1,\ldots,q_k,\bar{r}\}.
\]
For the clarity of further considerations by $\ell_{j,p}$ we will denote the fact corresponding to the satisfying valuation of the literal $L_{j,p}$, i.e.:
\[
\ell_{j,p}=
\begin{cases}
p_i &\text{when $L_{j,p}=x_i$,}\\
\bar{p}_i &\text{when $L_{j,p}=\neg x_i$.}
\end{cases}
\]

Now, we outline the interaction among the facts induced by the integrity constraints. We start with the simple observation that the FD $R:A_1\rightarrow B_1$ ensures that in every repair there is at most one fact corresponding to a valuation of each variable. In symbols:
\begin{equation*}
p_i\land \bar{p}_i\rightarrow \mathbf{false} \quad
\text{for $i\in\{1,\ldots,n+m\}$,}
\end{equation*}
The full TGD ensures that for every conjunct if the repair does not have a fact corresponding to a valuation satisfying the conjunct, then the fact $q_j$ is deleted or the fact $r$ is inserted:
\begin{equation*}
q_j\rightarrow \ell_{j,1}\lor \ell_{j,2} \lor \ell_{j,3} \lor r
\quad\text{for $j\in\{1,\ldots,k\}$.}
\end{equation*}
Inserting $r$ requires removing all the facts corresponding to valuations of the existential variables (the FD $R:A_2\rightarrow B_2$):
\begin{equation*}
p_i\land r \rightarrow \mathbf{false} 
\quad\text{and}\quad
\bar{p}_i\land r \rightarrow \mathbf{false}
\qquad \text{for $i\in\{n+1,\ldots,n+m\}$.}
\end{equation*}
It is important to note that this makes inserting $r$ quite a drastic way to repair the instance $I_\Psi$. Since $r$ does not belong to $I_\Psi$ and all facts corresponding to valuations of existential variables do, $\leq_I$-minimality ensures that such a way of repairing is considered only if for a given valuation of universal variables there does not exist a valuation of existential variables satisfying $\Psi$. Also, we observe that inserting $r$ requires deleting  $\bar{r}$, i.e. 
\begin{equation*}
r\land\bar{r}\rightarrow \mathbf{false}.
\end{equation*}
Consequently, the query used in the reduction checks if the fact $r$ is not deleted from any of the repairs:
\[
Q=\bar{r}.
\]

Figure~\ref{fig:fig1} contains the extended conflict hypergraph of $I_\Psi$ for 
\[
\Psi=\forall x_1,x_2,x_3\sepdot\exists x_4,x_5\sepdot (\neg x_1\lor x_4\lor x_2)\land (\neg x_2\lor\neg x_5\lor x_3).
\]
The dotted lines are used for stabilizing edges.
\begin{figure}[htb]
\begin{center}
\begin{tikzpicture}
\begin{scope}[xscale=1.4, yscale=1.1]
\node (pp1) at (0,0) {$p_1$};
\node (np1) at (1,0) {$\bar{p}_1$};
\draw (pp1) -- (np1);

\node (pp4) at (2,0) {$p_4$};
\node (np4) at (3,0) {$\bar{p}_4$};
\draw (pp4) -- (np4);

\node (pp2) at (4,0) {$p_2$};
\node (np2) at (5,0) {$\bar{p}_2$};
\draw (pp2) -- (np2);

\node (pp5) at (6,0) {$p_5$};
\node (np5) at (7,0) {$\bar{p}_5$};
\draw (pp5) -- (np5);

\node (pp3) at (8,0) {$p_3$};
\node (np3) at (9,0) {$\bar{p}_3$};
\draw (pp3) -- (np3);

\node (ps) at (4.5,1) {$r$};
\node (ns) at (4.5,2) {$\bar{r}$};
\draw (ps) -- (ns);

\draw (ps) -- (pp4);
\draw (ps) -- (np4);
\draw (ps) -- (pp5);
\draw (ps) -- (np5);

\node (nnp1) at (1,-1) {$\neg \bar{p}_1$};
\node (npp4) at (2,-1) {$\neg p_4$};
\node (npp2) at (4,-1) {$\neg p_2$};
\node (nps) at (4.5,-1) {$\neg r$};
\node (nnp2) at (5,-1) {$\neg \bar{p}_2$};
\node (nnp5) at (7,-1) {$\neg \bar{p}_5$};
\node (npp3) at (8,-1) {$\neg p_3$};

\draw[densely dotted] (nnp1) -- (np1);
\draw[densely dotted] (npp4) -- (pp4);
\draw[densely dotted] (npp2) -- (pp2);
\draw[densely dotted] (nps) -- (ps);
\draw[densely dotted] (nnp2) -- (np2);
\draw[densely dotted] (nnp5) -- (np5);
\draw[densely dotted] (npp3) -- (pp3);

\node (d1) at (3,-2) {$q_1$};
\node (d2) at (6,-2) {$q_2$};

\draw[blue!50!black] (3,-2) plot[smooth cycle, tension=0.85] coordinates{
+(0,-.4) +(-2.25,1.25) +(1.65,1.25)
};

\draw[blue!50!black] (6,-2) plot[smooth cycle, tension=0.85] coordinates{
+(0,-.4) +(2.25,1.25) +(-1.65,1.25)
};
\end{scope}
\end{tikzpicture}
\end{center}
\caption{\label{fig:fig1} $G(I_\Psi,F)$ for $\Psi=\forall x_1,x_2,x_3\sepdot\exists x_4,x_5\sepdot
(\neg x_1\lor x_4\lor x_2)\land (\neg x_2\lor\neg x_5\lor x_3)$.}
\end{figure}

The main claim of the reduction is :
\[
I_\Psi\models_F \bar{r} \iff \models \Psi,
\]
where $\models\Psi$ denotes that $\Psi$ is valid. 

For the {\em only if\/} part, we start by observing that no repair contains $r$. We also show that for any consistent instance $I_1\subseteq\{p_1,\bar{p}_1,\ldots,p_n,\bar{p}_n\}$ there exists a repair $I'$ such that $I'\leq_I I_1$ and such that $\{q_1,\ldots,q_k\}\subseteq I'$. Indeed, consider $J=I_1\cup\{q_1,\ldots,q_k,r\}$. Clearly, it is consistent and since no repair contains $r$, $J$ is not a repair. Consequently, there exists a repair $I'$ such that $I'<_I J$. It is easy to see that $I'$ is the required repair. 

Now, we take any valuation of universal variables $V_1$, construct the consistent instance 
\[
I_1=\{p_i\sepbar V_1(x_i)=\mathbf{true}, i\in\{1,\ldots,n\}\} \cup 
\{\bar{p}_i \sepbar V_1(x_i)=\mathbf{false}, i\in\{1,\ldots,n\}\},
\]
and take the repair $I'$ as described above. Naturally, for every $i\in\{1,\ldots,n\}$ either $p_i$ or $\bar{p}_i$ belongs to $I'$. The same holds for $i\in\{n+1,\ldots,n+m\}$ because $r\not\in I'$ and $I'$ is $\leq_I$-minimal. Consequently, the following valuation of the existential variables $x_{n+1},\ldots,x_{n+m}$ is properly defined
\[
V_2(x_i) = 
\begin{cases}
\mathbf{true} & \text{if $p_i\in I'$,}\\
\mathbf{false} & \text{if $\bar{p}_i\in I'$.}
\end{cases}
\]
We claim that $V_1\cup V_2\models \Phi$. Take any clause $C_j$ and observe that at least one of $\ell_{j,1},\ell_{j,2},\ell_{j,3}$ is present in $I'$ because $r\not\in I'$, $q_j\in I'$, and $I'$ is a repair. Consequently, $V_1\cup V_2$ assigns {\bf true} to at least one of $L_{j,1},L_{j,2},L_{j,3}$. 

Now, we show the {\em if\/} part by contradiction: we assume there exists a repair $I'$ such that $\bar{r}\not\in I'$ and we construct a consistent instance $I''$ such that $I''<_II'$.

First, we observe that $\bar{r}\not\in I'$ implies that $r\in I'$. By $\leq_I$-minimality this gives $\{q_1,\ldots,q_k\}\subseteq I'$. Also, for every $i\in\{1,\ldots,n\}$ either $p_i$ or $\bar{p}_i$ belongs to $I'$ and for $i\in\{n+1,\ldots,n+m\}$ neither $p_i$ or $\bar{p}_i$ belongs to $I'$. 

Consequently, the following valuation of the universal variables $x_1,\ldots,x_n$ is properly defined
\[
V_1(x_i) = 
\begin{cases}
\mathbf{true} & \text{if $p_i\in I'$,}\\
\mathbf{false} & \text{if $\bar{p}_i\in I'$.}
\end{cases}
\]
Since $\models \Psi$, there exists a valuation of the existential variables $V_2$ such that $V=V_1\cup V_2\models \Phi$. The instance $I''$ is defined as follows:
\begin{align*}
I''=&
\{\bar{r},q_1,\ldots,q_k\}\cup\{p_i \sepbar V(x_i)=\mathbf{true}, i\in\{1,\ldots,n+m\}\}\cup{}\\
&\{\bar{p}_i \sepbar V(x_i)=\mathbf{false},i\in\{1,\ldots,n+m\}\}.
\end{align*}
$V\models\Phi$ implies that $I''$ is consistent. Note that $I'$ and $I''$ agree on the facts $\{p_1,\bar{p}_1,\ldots,p_n,\bar{p}_n\}$, both contain $\{q_1,\ldots,q_k\}$ but $I''$ contains $\bar{r}$ and some of the facts $\{p_{n+1},\bar{p}_{n+1},\ldots,p_{n+m},\bar{p}_{n+m}\}$ whereas $I'$ contains none of them. This shows that $I''$ is relatively closer to $I$ than $I'$, i.e. $I'' <_I I'$. \qed
\end{proof}
\begin{corollary}
There exists a set of 2 FDs and one universal constraint for which repair checking is coNP-complete. 
\end{corollary}
\begin{proof}
We use the reduction from the proof of Theorem~\ref{thm:universal-cqa-intractable} to reduce $\mathcal{B}_F$ to the complement of 3SAT: a 3CNF $\Phi$ is treated as a $\forall^*\exists^*$QBF with no universally quantified variables. Let $F$ be the set of constraints as defined in the previous reduction and  $I_\Phi$ be the instance obtained from $\Phi$. We take $I'_\Phi=\{r,q_1,\ldots,q_k\}$ and claim that 
\[
\Phi\not\in\text{3SAT}\iff I_\Phi'\in Repairs(I_\Phi,F).
\]
The proof of this claim is analogous to the proof of Theorem~\ref{thm:universal-cqa-intractable}. 
\qed
\end{proof}

\section{Related work}
\label{sec:related}
Here we only discuss work relevant to our contributions and we refer the reader to surveys of the topic~\cite{Be06,BeCh03,Ch07,Fa08}.

In general, three different approaches to compute consistent query answers have been proposed: query rewriting, logic programming, and compact representation of repairs. Our work belongs to the last category. 

{\em Query rewriting} was the first approach proposed to compute consistent query answers~\cite{ArBeCh99}. A query $Q$ is rewritten into a query $Q'$ whose evaluation returns the set of consistent query answers to $Q$. An indisputable advantage of this approach is the ease of its incorporation into already existing applications. However, applicability of this approach is limited and certain conjunctive queries are known not to have rewritings~\cite{ChMa04,FuMi07,Wi07}.  

\cite{ArBeCh99} uses the notion of {\em residues} obtained from constraints to identify potential impact of integrity violations on the query results. The residues are used to construct rewriting rules for the atoms used in the query. This approach has been shown to be applicable to quantifier-free conjunctive queries in the presence of binary universal constraints. Chomicki and Marcinkowski~\cite{ChMa04} observe that if the set of constraints contains one FD per relation only, the conflict graph is a union of disjoint full multipartie graphs. This simple structure allows to construct rewriting for {\em simple} conjunctive queries, i.e conjunctive queries without repeated relation names and no variable sharing. The result of Chomicki and Marcinkowski has been further generalized by Fuxman and Miller~\cite{FuMi05,Fuxman,FuMi07} to allow restricted variable sharing (joins) in the conjunctive queries. The class $C_{forest}$ of allowed queries is defined using the notion of {\em join graph} of a query whose vertices are the literals used in the query and an edge runs from a literal $R_i$ to literal $R_j$ if there is a variable which occurs on a non-key attribute of $R_i$ and any attribute of $R_j$ (both occurrences have to be different if $i=j$). The class $C_{forest}$ consist of queries whose join graph is a forest, the joins are full and the join conditions are non-key to key. Wijsen~\cite{Wi07} presents a rewriting scheme for the class of {\em rooted} queries which further extends $C_{forest}$. We remark that the class of rooted queries is semantically defined and its subclass is captured with syntactic characterization using an alternative notion of the join graph.

Several approaches have been developed to compute consistent query answers using {\em logic programs} with disjunction and classical negation~\cite{ArBeCh03,BaBe03,EFGL03,GrGrZu01,GrGrZu03,NiVe02}. Essentially, all of them use disjunctive rules to model the process of repairing violations of constraints. In this way stable models of a program corresponds to the repairs of the inconsistent database. A query evaluated under the {\em cautious} semantics returns the answers present in every model, which naturally yields the consistent query answers. 

The main advantage of using this approach is its generality: typically arbitrary first-order (or even Datalog$^\neg$) queries are handled in the presence of universal constraints. Also, the repairing programs can be easily evaluated with existing logic program environments like {\sf Smodels} or {\tt dlv}~\cite{EFLW00}. We note, however, that the systems computing answers to logic programs usually perform grounding, which may be cost prohibitive if we are to work with large databases. Another disadvantage of this approach is that the class of disjunctive logic programs is known to be $\Pi_p^2$-complete. 

These difficulties are addressed in the INFOMIX system~\cite{EFGL03} with several optimizations geared toward effective execution of repairing programs. One is {\em localization} of conflicts with identification of the {\em affected database} which consists of all facts involved in constraint violations and all syntactically propagated {\em conflict-bound} facts (analogous to applying $T^*_F$). Another optimization involves using bit-vectors to encode fact membership to each repair and subsequent use of  bitwise aggregate function to find tuples that  present in every repair. This optimization, however, may be insufficient to handle databases with large numbers of conflicts because typically the number of repairs is exponential in the number of conflicts.  Recently, this deficiency has been addressed with {\em repair factorization}~\cite{EFGL08}. Essentially, the affected database is decomposed into parts that are conflict-disjoint (no two mutually conflicting facts are in separate parts). When computing consistent answers to a query only parts that are simultaneously spanned by the query are considered at a time. We observe an analogy to computing consistent query answers using hypergraphs: when finding whether {\bf true} is the consistent answer to a ground query Algorithm~\ref{alg:cqa} analyzes base fragments of repairs obtained by combining the hyperedges adjacent to facts from the query. 

Our work was inspired by positive results for denial constraints presented in~\cite{ChMa04}. There, the repairs are obtained by  deleting facts only and consequently the repairs are subsets of the original instance. \cite{ChMa04} also investigates using {\em subset} repairs obtained to define consistent query answers in the presence of inclusion dependencies (IND), i.e. formulas of the form  
\[
\forall \bar{x}_1\exists \bar{x}_3\sepdot R(\bar{x}_1)\rightarrow P(\bar{x}_2,\bar{x}_3),
\]
where $\bar{x}_2\subseteq\bar{x}_1$. An IND of this form is commonly written $R[X]\subseteq P[Y]$, where $X$ and $Y$ are the sets of attributes corresponding to $\bar{x}_2$ in $R$ and $P$ respectively. An IND $R[X]\subseteq P[Y]$ is a {\em foreign-key} dependency if $Y$ is the {\em key} of $P$. We note that universal constraints capture only {\em full} INDs, i.e. INDs with no existentially quantified variables. \cite{ChMa04} shows that consistent query answering is in PTIME for quantifier-free and simple conjunctive queries in the presence of foreign-key dependencies and one key dependency per relation. It is also shown that relaxing the restriction on the set of integrity constraints leads to intractability and consistent query answering for arbitrary sets of INDs and FDs becomes $\Pi^p_2$-complete.

Using subset repairs is natural in scenarios like {\em data warehousing}, where the data is complete but may be incorrect. In particular we can assume that if a fact is not present in the original database, then it is not true. Obtaining repairs by deletion of facts only is not necessarily a natural approach in the scenarios where we cannot assume that information missing in the database is false, for instance in the context of integration of sources that may be missing some information. Then, we might want to consider standard repairs obtained by deleting and inserting a minimal set of facts, i.e. repairs in the sense of Definition~\ref{def:repair}. We observe, however, that while in the case of universal constraints the missing facts that create conflicts are implicitly defined, the presence of existentially quantified variables in INDs leads to possibly infinite number of repairs. 

Cali et al. in~\cite{CaLeRo03} show that consistent query answering becomes undecidable for arbitrary sets of INDs and FDs. The problem becomes decidable when the set of integrity constraints is restricted to {\em non-key-conflicting} INDs; IND $R[X]\subseteq P[Y]$ is non-key-conflicting if $Y$ is not a strict superset of the key of $P$. Then, the problem of consistent query answering is $\Pi_p^2$-complete.  

Another compact representation of all repairs is {\em nucleus} \cite{Wij03,Wij05}. In this approach all repairs are represented by a tableau (a table with free variables), and queries are evaluated in the standard way (answers with variables are discarded). We note that for some classes of constraints, constructing the nucleus may take an exponential time to complete.

\rrmark{Updated to a newer version of the paper.}\rrtext{\cite{AfKo09} provides an thorough study of the complexity of repair checking for 4 different notions of minimality used to define repairs: minimality of symmetric set difference (Definition~\ref{def:repair}), minimality of asymmetric set difference (which yields subset repairs), minimality of the cardinality of symmetric set difference, and minimality of the cardinality of symmetric difference on every relation. The classes of the considered integrity constraints include denial constraints, inclusion dependencies, equality-generating dependencies, and {\em weakly acyclic} and {\em local-as-view} (LAV) tuple-generating dependencies. The results offer additional insight into the problem of repair checking in the presence of full TGDs for the symmetric and asymmetric set difference notion of minimality. The authors show that for full TGDs repair checking is PTIME-hard, which in view of Theorem~\ref{thm:rcheck-cyclic} makes the problem PTIME-complete. 
It is a general belief, based on the assumption $\text{NC}\subsetneq\text{PTIME}$, that there do not exist fast parallel algorithms for PTIME-complete problems. This suggest that database repairing (Algorithm~\ref{alg:construct-repair}) using a parallel computation model does not guarantee an efficiency improvement. However, the authors also show that for weakly acyclic LAV tuple-generating dependencies the problem is in LOGSPACE (which is included in NC). On the other hand, repair checking easily becomes coNP-complete if we relax the restrictions on the set of constraints; for instance, if we consider weakly acyclic TGDs without the LAV restriction. We remark that these results are orthogonal to Theorem~\ref{thm:rcheck-cyclic} as the classes of constraints are incomparable.
}

\section{Conclusions and future work}
\label{sec:future}
In this paper we investigated the complexity of computing consistent query answers in the presence of universal constraints. We proposed an extended version of the conflict hypergraph. Its size is polynomial in the size of the database and it captures all repairs w.r.t. to the given set of universal constraints. Hence, we consider it to be a compact representation of all repairs. This property is essential for using the extended conflict hypergraph to compute consistent query answers.   

Extending the notions of conflicts to include negations of facts leads, however, to a significant increase of computational complexity. Consistent query answering is $\Pi_P^2$-complete in the presence of universal constraints even when using atomic queries. The problem becomes coNP-complete when we restrict the set of constraints to contain full tuple-generating dependencies and denial constraints only; then the conflicts can contain at most one negation of a fact. If we further restrict the integrity constraints to join dependencies, denial constraints, and acyclic full tuple-generating dependencies, then the problem of consistent answering becomes tractable for quantifier-free queries. Consequently, we present an extension of the algorithm of Chomicki and Marcinkowski~\cite{ChMa04} that finds if {\bf true} is the consistent answer to a closed quantifier-free query. 

The problem of repair checking is also intractable for universal constraints. It becomes tractable if we restrict the constraints to full tuple-generating dependencies and denial constraints. Consequently, we present a polynomial repairing algorithm. It is both sound (always produces a repair) and complete (every repair can be produced). 

The summary of computational complexity results is presented in Table~\ref{tab:summary-universal}; its last row is taken from \cite{ChMa04}.

\begin{table}[htb]
\begin{center}
\begin{tabular}{|c|c|c|c|}
\hline
\multirow{2}{*}{Constraints}&
\multirow{2}{20mm}{\parbox{20mm}{\begin{center}Repair Checking\end{center}}}&
\multicolumn{2}{c|}{Consistent Answering to}\\
\cline{3-4}
&
&$\{\forall,\exists\}$-free queries
&conjunctive queries\\
\hline
Universal & coNP-complete & 
\multicolumn{2}{c|}{$\Pi^p_2$-complete}\\
\hline
Full TGDs + Denial & PTIME &
\multicolumn{2}{c|}{coNP-complete}\\
\hline
\multirow{2}{34mm}{\parbox{34mm}{\begin{center}Acyclic full TGDs + Denial + JDs\end{center}}}&
\multirow{2}{*}{PTIME}&
\multirow{2}{*}{PTIME}&
\multirow{2}{*}{coNP-complete}\\
&&&\\
\hline
Denial & PTIME & PTIME & coNP-complete\\
\hline
\end{tabular}
\end{center}
\caption{\label{tab:summary-universal}Summary of complexity results for universal constraints.}
\end{table}

We envision several possible directions of future study. First, we would like to investigate practical applicability of our approach. The main obstacle lays in the high degree of the polynomials used to bound the number of of all supports and blocks, and consequently in the high degree of the polynomial describing the complexity of Algorithm~\ref{alg:cqa}. We observe that the bounds are estimates of the pessimistic case where the number of conflicts (ground rules) in the database is very high and the set of integrity constraints very complex. We believe that in practical scenarios the amount of conflicts is small enough to be stored in the main memory and the acyclic height of the set of integrity constraints rather small. 

\rrmark{Added handling conjunctive queries to Future Work.}\rrtext{
Although computing consistent answers to arbitrary conjunctive queries is long known to be intractable~\cite{ChMa04}, considerable effort has been made to find practical subclasses of conjunctive queries for which consistent answering is tractable~\cite{Wi07,FuMi07,GLRR05}. Usually, tractability comes at the price of restricting the class of constraints to primary key constraints. However, it would be interesting to see for what subclasses of universal constraints similar techniques could be used to handle conjunctive queries. Another interesting challenge in this direction is a generalization of Algorithm~\ref{alg:cqa} to handle sets of universal constraints and arbitrary queries with quantifiers. Because of the negative complexity results, we cannot expect that a generalized algorithm would work in polynomial time (unless P=NP). We believe, however, that in most practical cases such an algorithm should not require exponential time. This belief is based on the promising results of heuristics used to optimize the INFOMIX system~\cite{EFGL03,EFGL08} and its conceptual closeness to Algorithm~\ref{alg:cqa} (see Section~\ref{sec:related}).
}

It would be interesting to see if using an alternative definition of repairs would affect the complexity of consistent query answering and repair checking in the presence of universal constraints. For instance, we observe that if we consider repairs obtained by deleting facts only, then the repairs are maximal consistent subsets of the original instance. It would seem that this property simplifies reasoning about repairs allowing to employ algorithms similar to those used for denial constraints (where all repairs are obtained by deleting facts only). For example, a subset $I'$ of an instance $I$ is a repair of $I$ w.r.t. to a set of denial constraints if and only if $I'\cup\{R(t)\}$ is inconsistent for any $R(t)\in I\setminus I'$. This is not necessarily true in the case of universal constraints, where we need to check that $I'\cup X$ is inconsistent for every nonempty $X\subseteq I\setminus I'$. In fact, our preliminary research shows that this problem remains coNP-complete. Also, we believe that the positive results carry to the setting of subset repairs as well.
\bibliographystyle{plain}
\bibliography{general}
\appendix
\section{Omitted proofs}
\label{app:omitted-proofs}
\noindent
{\em
{\bf Proposition~\ref{prop:one-jd-per-relation}}
For any set of JDs $\{jd_1,\ldots,jd_n\}$ on the same relation there exists a JD $jd^*$ such that an instance satisfies $\{jd_1,\ldots,jd_n\}$ if and only if it satisfies $jd^*$.
}

\begin{proof}
The proof is by induction over $n$. For $n=0$ we note that every relation satisfies the trivial JD: $R\jd[attrs(R)]$, where $attrs(R)$ is the set of all attributes of $R$. The hypothesis is also trivially satisfied for $n=1$. To show the inductive step it suffices to show that two JDs $jd_1=R\jd[X_1,\ldots,X_n]$ and $jd_2=R\jd[Y_1,\ldots,Y_m]$ are equivalent to 
\rrmark{Changed formatting.}
\[
jd^*=R\jd[X_1\cap Y_1,\ldots, X_1\cap Y_m,\ldots,X_n\cap Y_1,\ldots,X_n\cap Y_m]. 
\]
To prove this equivalence we use standard relation algebra~\cite{AbHuVi95} and recall that a JD $R\jd[Z_1,\ldots,Z_k]$ is defined as $R=\pi_{Z_1}(R)\Join\ldots\Join\pi_{Z_k}(R)$. 

First we note that every instance satisfies $\pi_{Z_1}(R)\Join\ldots\Join\pi_{Z_k}(R)\subseteq R$ for any sets $Z_1,\ldots,Z_k$ of attributes of $R$ whose union is $attrs(R)$. Hence, it suffices to show that 
\[
R\subseteq \pi_{X_1\cap Y_1}(R)\Join\ldots\pi_{X_1\cap Y_m}(R)\Join\ldots\Join \pi_{X_n\cap Y_1}(R)\Join\ldots\pi_{X_n\cap y_m}(R)
\] in any instance $I$ that satisfies $jd_1$ and $jd_2$. 

We fix an instance $I$ and let $r$ be the set of all tuples that belong to the relation $R$ in $I$. Take any $t\in r$ and let $t_{X_i}=\pi_{X_i}(t)$ , $t_{Y_j}=\pi_{Y_j}$, 
and $t_{X_i\cap Y_j}=\pi_{X_i\cap Y_j}(t)$ for any $i\in\{1,\ldots,n\}$ and any $j\in\{1,\ldots,m\}$. Since $I$ satisfies 
$jd_1$ and $jd_2$, 
we have that $t=t_{X_1}\Join\ldots\Join t_{X_n}\in r$ and $t=t_{Y_1}\Join\ldots\Join t_{Y_m}\in r$. 
Now, we observe that 
\begin{equation}
t_{X_1\cap Y_1}\Join\ldots\Join t_{X_1\cap Y_m}\Join\ldots\Join t_{X_n\cap Y_1}\Join\ldots\Join t_{X_n\cap Y_m} =t_{X_1}\Join\ldots\Join t_{X_n}.\tag*{\qed}
\end{equation}
\end{proof}

We recall that the lhs of a ground rule is represented with a set of facts obtained from grounding the atoms of some constraint. If the same fact is obtained by grounding more than one atom in the constraint, then it is not repeated in the ground rule. We continue to use this representation, but on some occasions we will require to know the duplicates in ground JD rules. Then, the lhs is represented with a bag rather than a set, and we call such a rule {\em unfolded}. Naturally, every rule can be unfolded, although not always unambiguously. We remark that this ambiguity does not affect the correctness our considerations, and hence we ignore it.
\begin{example}[Unfolded rule]\label{ex:unfolded-rule}
Suppose a schema consisting of one relation name $R(A,B,C,D)$ and let the set of integrity constraints contain one JD $R\jd[AB,BC,CD]$ which is represented with the following formula
\[
R(x,y,z,s)\land 
R(x',y,z',s')\land
R(x'',y'',z',s'')
\rightarrow
R(x,y,z',s'').
\]
The following ground rule is one of possible instantiations of the formula above.
\[
R(0,0,0,0)\land R(1,0,0,1)\rightarrow R(0,0,0,1).
\]
Its unfolding is 
\[
R(0,0,0,0)\land R(1,0,0,1)\land R(1,0,0,1)
\rightarrow R(0,0,0,1).
\]
\end{example}

{\em{\bf Lemma~\ref{lemma:jd-weaving}}
If $Rules(I,F)$ contains the following two unfolded ground rules
\[
r'=R(t'_1)\land\ldots\land R(t'_k)\jdrightarrow R(t_i)
\quad\text{and}\quad
r''=R(t_1)\land\ldots\land R(t_k)\jdrightarrow R(t)
\]
for some $i\in\{1,\ldots,k\}$, then there exists $j\in\{1,\ldots,k\}$ such that $Rules(I,F)$ contains also 
\[
r^*=R(t_1)\land\ldots\land R(t_{i-1})\land R(t'_j)\land R(t_{i+1})\land\ldots\land R(t_k)\jdrightarrow R(t).
\]
}

\begin{proof}
We assume that the rules are obtained from grounding the join dependency $R\jd[X_1,\ldots,X_k]$ and recall that it is represented as the following full TGD:
\[
R(\bar{x}_1)\land\ldots\land R(\bar{x}_k) \land
\bigwedge_{1\leq \alpha,\beta \leq k} 
\bar{x}_\alpha[X_\alpha\cap X_\beta]=\bar{x}_\beta[X_\beta\cap X_\alpha]
\rightarrow R(\bar{y}),
\]
where $\bar{y}\subseteq \bar{x}_1\cup\ldots\cup\bar{x}_n$ such that $\bar{y}[X_\alpha\setminus\bigcup_{1\leq \beta < \alpha} X_\beta] = 
\bar{x}_\alpha[X_\alpha\setminus\bigcup_{1\leq \beta < \alpha} X_\beta]$ for $\alpha\in\{1,\ldots,k\}$. W.l.o.g. we assume that the order of facts in the lhs of rules $r'$ and $r''$ corresponds to the order of the atoms in the definition of the constraints. With this assumption a rule $R(s_1)\land\ldots\land R(s_k)\jdrightarrow R(s)$ belongs to $Rules(I,F)$ if and only if 
\begin{align}
s_\alpha[X_\alpha]&=s[X_\alpha]&
&\text{for every $\alpha\in\{1,\ldots,k\}$}\label{eq:jd1}\\
s_\alpha[X_\alpha\cap X_\beta]&=s_\beta[X_\alpha\cap X_\beta]&
&\text{for every $\alpha,\beta\in\{1,\ldots,k\}$.}\label{eq:jd2}
\end{align}
Also, with the assumption on the order in which the facts are listed in rules, we show that the claim of the lemma holds for $j=i$. 

$t_i'[X_i]=t_i[X_i]=t[X_i]$ follows from \eqref{eq:jd1} for $r''$ and $r'$ (for $\alpha=i$). Since $X_i\cap X_\beta\subseteq X_i$, we get $t_i'[X_i\cap X_\beta]=t_i[X_i\cap X_\beta]$, and from \eqref{eq:jd2} for $r'$ (for $\alpha=i$) we obtain $t_i'[X_i\cap X_\beta]=t_\beta[X_i\cap X_\beta]$ for every $\beta\in\{1,\ldots,k\}$. The remaining equations needed to prove $r^*$ follow trivially from \eqref{eq:jd1} and \eqref{eq:jd2} for $r''$.
\qed
\end{proof}

\noindent
{\em{\bf Proposition~\ref{prop:jd-support}}
For every $I'\in Repairs(I,F)$ and every $R(t)\in Hull(I,F)$
\[
R(t)\in I'\iff \exists S\in Supp(R(t))\sepdot S\subseteq I'.
\]
}

\begin{proof}
We fix a repair $I'$ and we say that an $\ell$-support $S$ of $R(t)\in Hull(I,F)$ is {\em proper} if $S\subseteq I'$. 

The {\em if\/} part is proved with a simple induction over the depth of the derivation of a support. The proof of the {\em only if\/} is based on the following simple idea. A fact $R(t)\in I'\setminus I$ is present in the repair $I'$ to satisfy some ground (full TGD) rule. We identify this ground rule by considering an inconsistent instance $I'\setminus\{R(t)\}$. We use this rule to show that $R(t)$ has a proper support. 

Recall that the acyclic height $height(R)$ in $\mathcal{D}(\mathcal{S},F)$ of a relation name $R$ is the maximal length of a directed acyclic path that begins in $R$. For $\ell\in\{-1,0,\ldots,h\}$ define  $\mathcal{S}^\ell=\{R\in\mathcal{S}\sepbar height(R)\leq\ell\}$ and note that $\mathcal{S}^{h}=\mathcal{S}$. 

We show with induction over $\ell\in\{-1,0,\ldots,h\}$ that for $R\in \mathcal{S}^\ell$
\[
R(t)\in I' \Rightarrow \exists S\in Supp^\ell(R(t))\sepdot S\subseteq I'.
\]
For $\ell=-1$ the claim is trivially true because $\mathcal{S}^{-1}=\varnothing$. We prove the inductive step by taking the set of all facts in $I'$ that do not have a proper support: 
\[
T=\{R(t)\in I' \sepbar R\in \mathcal{S}^\ell\land 
\nexists S\in Supp^\ell(R(t))\sepdot S\subseteq I'\}.
\]

We observe that $T\subseteq I'\setminus I$ because all facts that belong to $I$ have a proper support obtained with the rule $\mathtt{S}_0^{-1}$. Also, by IH $T$ contains no facts using a relation name from $\mathcal{S}^{\ell-1}$, i.e. $T$ contains only facts using relations names whose acyclic height is $\ell$. 

Now, we show that $I'\setminus T$ is consistent. As a subset of a consistent instance it satisfies all denial constraints. Hence, we only need to check satisfiability of ground full TGD rules having in the rhs a relation name of acyclic height $\ell$.

First, we take a non-JD ground rule 
\[
R_1(t_1)\land\ldots\land R_n(t_n)\rightarrow R(t)
\]
such that $R_i(t_i)\in I'\setminus T$ for all $i\in\{1,\ldots,n\}$. Note that every $R_i\in\mathcal{S}^{\ell-1}$ and by IH every $R_i(t_i)$ has an $(\ell-1)$-support $S_i$ such that $S_i\subseteq I'$. Now, $S=S_1\cup\ldots\cup S_n\subseteq I'$ is an $\ell$-support of $R(t)$ constructed with the rule $\mathtt{S}_1^\ell$. Consequently, $R(t)\not\in T$. Naturally, $R(t)\in I'$ because $I'$ is consistent. 

Now, consider an (unfolded)  JD rule 
\[
r^*_0=R(t_1)\land\ldots\land R(t_k)\jdrightarrow R(t)
\] 
such that every $R(t_i)\in I'\setminus T$. 
Repeatedly we use Lemma~\ref{lemma:jd-weaving} and instances of $\mathtt{S}_1^\ell$, used to obtain proper $\ell$-supports of $R(t_i)$'s, to show the existence of the following elements:
\begin{itemize}
\item a ground rule $r=R(t_1^*)\land\ldots\land R(t_n^*)\jdrightarrow R(t)$ (being the last element of a constructed sequence of (unfolded) ground rules $r_0^*,\ldots,r_k^*=r$),
\item ground rules $r_p^*=R_{p,1}(t_{p,1}^*)\land\ldots\land R_{p,n_p}(t_{p,n_p}^*)\rightarrow R(t_p^*)$ for every $p\in\{1,\ldots,k\}$,
\item proper $(\ell-1)$-supports $S_{p,q}^*$ of $R_{p,q}(t_{p,p}^*)$ for every $p\in\{1,\ldots,k\}$ and $q\in\{1,\ldots,n_p\}$.
\end{itemize}
For $p\in\{1,\ldots,k\}$, let the $\ell$-support of $R(t_p)$ be obtained with the following instance of $\mathtt{S}_1^\ell$:
\[
\genfrac{}{}{0.5pt}{}{\begin{aligned}
&&&r_p'=R(t_1')\land\ldots\land R(t_{k_p}')\jdrightarrow R(t_p)\\
&&&r_{p,\alpha}'=R_{\alpha,1}(t_{\alpha,1}')\land\ldots\land R_{\alpha,n_\alpha}(t_{\alpha,m_\alpha}')\rightarrow R(t_\alpha')
\quad\forall \alpha\in\{1,\ldots,k_p\}\\
&R(t_p)\not\in I& &S_{\alpha,\beta}\in Supp^{\ell-1}(R_{\alpha,\beta}(t_{\alpha,\beta}'))\quad
\forall \alpha\in\{1,\ldots,k\},\;\forall \beta\in\{1,\ldots,m_\alpha\}\end{aligned}}{
\bigcup_{\alpha,\beta} S_{\alpha,\beta}\in Supp^\ell(R(t_p))
}
\]
We apply Lemma~\ref{lemma:jd-weaving} to 
\[
r_p'=R(t_1')\land\ldots\land R(t_{k_p}')\jdrightarrow R(t_p)
\]
and
\[
r^*_p=R(t_1^*)\land\ldots\land R(t_{p-1}^*)\land R(t_p)\land R(t_{p+1})\land\ldots\land R(t_k)\jdrightarrow R(t)
\]
to obtain 
\[
r^*_p=R(t_1^*)\land\ldots\land R(t_{p-1}^*)\land R(t_p^*)\land R(t_{p+1})\land\ldots\land R(t_k)\jdrightarrow R(t), 
\]
where $R(t_p^*)=R(t_j')$ for some $j\in\{1,\ldots,k_p\}$ indicated by Lemma~\ref{lemma:jd-weaving}. From the instance of $\mathtt{S}_1^\ell$, for $r_p^*$ we take $r_{p,j}'$, and for $S_{p,q}^*$ we take $S_{j,q}$ for every $q\in\{1,\ldots,m_j\}$.

The elements above show that $R(t)$ has a proper support, i.e. $R(t)\not\in T$. Again, $R(t)\in I'$ because $I'$ is consistent. This finishes the proof that $I'\setminus T$ is consistent. 

Now, recall that $T\subseteq I'\setminus I$, and therefore $I'\setminus T \leq_I I'$. Since $I'$ is $\leq_I$-minimal consistent instance, $I'\setminus T= I'$, and consequently $T=\varnothing$.\qed
\end{proof}

\noindent
{\em
{\bf Proposition~\ref{prop:jd-block}}
For every $I'\in Repairs(I,F)$ and every $R(t)\in Hull(I,F)$ 
\[
R(t)\not\in I'\iff \exists (B,N)\in Block(R(t))\sepdot B\subseteq I' \land N\cap I'=\varnothing.
\]
}

\begin{proof}
We fix a repair $I'$ and we say that an $\ell$-block $(B,N)$ of $R(t)\in Hull(I,F)$ is {\em proper} if $B\subseteq I'$ and $N\cap I'=\varnothing$.

The {\em if\/} part is proved with a simple induction over the depth of derivation of a block. The proof of the {\em only if\/} part, although technically complex, is based on the following simple idea. $R(t)\in I\setminus I'$ is absent in the repair $I'$ because its presence would cause a violation of some ground rule. We identify this rule by considering an inconsistent instance $I'\cup\{R(t)\}$. We use this rule to show that $R(t)$ has a proper block. 

Recall that the acyclic depth $depth(R)$ in $\mathcal{D}(\mathcal{S},F)$ of a relation name $R$ is the maximal length of a directed acyclic path that ends in $R$. For $\ell\in\{-1,0,\ldots,h\}$ we define $\mathcal{S}^\ell=\{R\in\mathcal{S}\sepbar depth(R)\leq \ell\}$. Note that $\mathcal{S}^h=\mathcal{S}$ as the acyclic depth of every relation name is bounded by the acyclic height of $\mathcal{D}(\mathcal{S},F)$. 

We show with an induction over $\ell$ that for $R\in\mathcal{S}^\ell$
\[
R(t)\not\in I'\Rightarrow \exists (B,N)\in Block^\ell(R(t))\sepdot B\subseteq I'\land N\cap I'=\varnothing. 
\]

For $\ell=-1$ the claim is trivially true because $\mathcal{S}^{-1}=\varnothing$. We prove the inductive step by contradiction: we assume that the set of facts that are not in $I'$ and that do not have a proper $\ell$-block,
\[
T=\{R(t)\in Hull(I,F)\setminus I' \sepbar  R\in\mathcal{S}^l\land\text{$R(t)$ has no proper $\ell$-block}\}
\]end{align*}
is nonempty. We note that $T\subseteq I\setminus I'$ because any fact $R(t)\in Hull(I,F)\setminus I$ has a proper block obtained with the rule $\mathtt{B}_0^{-1}$.

Now, take any element $R(t)\in T$, let $jd_R\in F$ be the JD on relation $R$, and consider the set $V=T_{\{jd_R\}}(I'\cup\{R(t)\})\setminus I'$. For any $R(t')\in V$ by $r_{R(t')}$ we denote an arbitrarily chosen ground JD rule $R(t)\land R(t_1)\land\ldots R(t_k)\jdrightarrow R(t')$ such that every $R(t_i)\in I'$.

We claim that: (1) $I'\cup V$ is consistent, and (2) $V\subseteq T$. 
\begin{enumerate}
\item[(1)]  Because $I'\cup V$ is obtained by adding facts using the relation name $R$ to a consistent instance, we only need to verify that ground rules having $R$ in their lhs are satisfied. 

First, take an (unfolded) ground JD rule 
\[
r''=R(t_1')\land\ldots\land R(t_p')\land R(t_1)\land\ldots\land R(t_k)\jdrightarrow R(s)
\] 
such that every $R(t'_i)$ belongs to $V$ and every $R(t_i)$ belongs to $I'$. We iteratively apply Lemma~\ref{lemma:jd-weaving} to the rule above using rules $r_{R(t'_i)}$ and obtain the rule 
\[
R(t_1^*)\land\ldots\land R(t_p^*)\land R(t_1)\land\ldots\land R(t_k)\jdrightarrow R(s),
\] 
where each $R(t_i^*)$ is either $R(t)$ or belongs to $I'$. If all $R(t_i^*)$'s belong to $I'$, then $R(s)\in I'$ because $I'$ is consistent. Otherwise, $R(s)\in V$. 

For the remaining (non-JD) ground rules, we show that if such a rule is not satisfied in $I'\cup V$, then a proper $\ell$-block for $R(t)$ can be constructed (which contradicts $R(t)\in§ T$). If there is a ground denial rule 
\[
R(t_1')\land\ldots\land R(t_n')\land R_1(s_1)\land\ldots\land R_m(s_m)\rightarrow\mathbf{false}
\]
such that every $R(t_i')\in V$ and every $R_j(s_j)\in I'$, then a proper $\ell$-block of $R(t)$ is constructed with the rule $\mathtt{B}_1^{-1}$. Similarly, if there is a ground rule
\[
R(t_1')\land\ldots\land R(t_n')\land R_1(s_1)\land\ldots\land R_m(s_m)\rightarrow P(s)
\]
such that every $R(t_i')\in V$, every $R_j(s_j)\in I'$, and $P(s)\not\in I'\cup V$, then by IH $P(s)$ has a proper $(\ell-1)$-block and a proper $\ell$-block is constructed with the rule $\mathtt{B}_2^\ell$. 
\item[(2)] For $R(t')\in V$ we show that if $R(t')$ has a proper $\ell$-block, then $R(t)$ has a proper $\ell$-block as well (which contradicts $R(t)\in T$).

Suppose some $R(t')\in V$ has a proper $\ell$-block constructed with the following instance of the rule $\mathtt{B}_2^\ell$:
\[
\genfrac{}{}{0.5pt}{}{\begin{aligned}
&&&r=R(t_1')\land\ldots\land R(t_n')\land 
R_1(s_1)\land\ldots\land R_m(s_m)\rightarrow P(s)\\
&&&r_i=R(t')\land R(t_{i,1}')\land\ldots\land R(t_{i,k_i}')\jdrightarrow R(t_i')
\quad\forall i\in\{1,\ldots,n\}\\
&&&S_{i,j}\in Supp(R(t_{i,j}'))
\quad\forall i\in\{1,\ldots,n\},\;\forall j\in\{1,\ldots,k_i\}\\
&&& S_p\in Supp(R_p(t_p))\quad\forall p\in\{1,\ldots,m\}\\
&R(t')\in I& &(B,N)\in Block^{\ell-1}(P(s))
\end{aligned}
}{
(\bigcup_{i,j} S_{i,j}\cup\bigcup_p S_p\cup B,N)\in Block^\ell(R(t'))
}
\]
For every $i\in\{1,\ldots,n\}$ we apply Lemma~\ref{lemma:jd-weaving} to $r'=r_{R(t')}$ and $r''=r_i$, and obtain $r_i^*$.\footnote{More precisely, we take an unfolded version of $r''$ and apply Lemma~\ref{lemma:jd-weaving} to every occurrence of $R(t')$ in $r''$.} We observe that for every $i\in\{1,\ldots,n\}$ if the ground rule $r_i^*$ does not have $R(t)$ in its lhs, then all the facts in its lhs belong to $I'$ and consequently have a proper support. Let $X\subseteq \{1,\ldots,n\}$ be the set of indexes of all rules $r_i^*$ which have $R(t)$ in their lhs. Using $\mathtt{B}_1^\ell$ with the ground rules $r$, $r_i^*$ for $i\in X$, the corresponding proper supports, and the proper $(\ell-1)$-block of $P(s)$ we obtain a proper $\ell$-block of $R(t)$. 
\end{enumerate}
Finally, we observe that $I'\cup V<_I I'$ because $V\subseteq T \subseteq I\setminus I'$. However, $I'$ is by definition a $\leq_I$-minimal consistent instance; a contradiction. \qed
\end{proof}

\noindent
{\em
{\bf Proposition~\ref{prop:jd-support-and-block-size}}
For any fact $R(t)$ the sets $Supp(R(t))$ and $Block(R(t))$ can be constructed in time polynomial in the size of $I$. 
}

\begin{proof}
Here we only give a combinatorial argument showing that the number of all supports and blocks of a fact is bounded by a polynomial of the size of $I$. A polynomial algorithm that generates the supports and blocks can be easily derived. 

By $K$ we denote the maximum number of atoms used in the definition of a constraint in $F$. Since we assume the set of constraints to be fixed, $K$ is a constant. Also, note that the acyclic height $h$ of the dependency graph $\mathcal{D}(\mathcal{S},F)$ is bounded by the cardinality of $\mathcal{S}$. 

We recall that every ground rule corresponds to a subset of $Hull(I,F)$ of cardinality at most $K$. Hence, the number of all ground rules is bounded by $R=|Hull(I,F)|^{K+1}$.

First, with a simple induction over $\ell\in\{-1,0,\ldots,h\}$ we show that the number of all $\ell$-supports of a fact is bounded by $N=R^{(K+1)^{2(\ell+1)}}$. This bound holds trivially for $\ell=-1$. To show the inductive step, we note that $\mathtt{S}_1^\ell$ can be instantiated with at most $R^{K+1}$ possible combinations of ground rules and $(\ell-1)$-supports of $K^2$ facts. By IH each of the facts has at most  $R^{(K+1)^{2\ell}}$ $(\ell-1)$-supports. Together, the number of possible combinations is bounded by 
\begin{align*}
R^{K+1}\left(R^{(K+1)^{2\ell}}\right)^{K^2}=&
R^{(K+1)^{2\ell}K^2+K+1}=
R^{(K+1)^{2\ell}((K+1)^2-2K-1)+K+1}\\
=&
R^{(K+1)^{2(\ell+1)}-(2K+1)(K+1)^{2\ell}+K+1}\leq
R^{(K+1)^{2(\ell+1)}}.
\end{align*}
Hence the number of all supports of a fact is bounded by $N=R^{(K+1)^{2(h+1)}}$.

Now, with a simple induction over $\ell\in\{-1,0,\ldots,h\}$ we show that the number of all $\ell$-blocks of a fact is bounded by $(R^{K+1}N^{K^2})^{\ell+2}$. For $\ell=-1$ the set of supports is constructed either with $\mathtt{B}_0^{-1}$ or $\mathtt{B}_1^{-1}$. For the former the claim holds trivially. For the latter we observe that there are at most $R^{K+1}$ combinations of ground rules and $N^{K^2}$ combinations of supports used in $\mathtt{B}_1^{-1}$, giving together $R^{K+1}N^{K^2}$. Similarly, for the inductive step we observe that $\mathtt{B}_2^\ell$ can be instantiated with $R^{K+1}$ combinations of rules, $N^{K^2}$ combinations of supports, and, from IH, at most $(R^{K+1}N^{K^2})^{\ell+1}$ $(\ell-1)$-blocks. Together, this gives us exactly $P=(R^{K+1}N^{K^2})^{\ell+2}$.

Finally, we observe that both $N$ and $P$ are polynomials when viewed as functions of the size of $I$. 
\qed
\end{proof}
\end{document}